\documentclass[10pt]{article}

\usepackage{amsmath}
\usepackage{times}
\usepackage{bm}
\usepackage{natbib}
\usepackage{amsfonts}

\usepackage{enumerate}
\usepackage{mathscinet}

\usepackage{graphicx}
\usepackage[margin=1cm, textfont=small]{caption}

\usepackage[textwidth=1.5cm]{todonotes}
\usepackage{url}

\usepackage{booktabs,caption}
\usepackage[flushleft]{threeparttable}
\usepackage{setspace}

\usepackage{listings}

\allowdisplaybreaks

\usepackage{xcolor}
\usepackage{pdfpages}
\usepackage{lmodern}
\usepackage[utf8]{inputenc} 
\usepackage[T1]{fontenc}
\usepackage{epsfig}
\usepackage{blindtext}
\usepackage{wrapfig}
\usepackage{float}
\usepackage{url}
\usepackage{eurosym}
\usepackage{amsbsy}
\usepackage{amssymb}
\usepackage{mathtools}
\usepackage{appendix}
\usepackage{esvect} 
\usepackage{algorithm,algpseudocode}
\newcounter{algsubstate}
\renewcommand{\thealgsubstate}{\alph{algsubstate}}
\newenvironment{algsubstates}
  {\setcounter{algsubstate}{0}%
   \renewcommand{\State}{%
     \stepcounter{algsubstate}%
     \Statex {\footnotesize\thealgsubstate:}\space}}
  {}

\usepackage{mathbbol}
\usepackage{relsize}
\usepackage{listings}
\usepackage{todonotes}

\newcommand{\norm}[1]   {\left\| #1 \right\|}
\newcommand{\abs}[1]    {\left| #1 \right|}

\newcommand{\GL}{{\sc glasso}}
\newcommand{\DPGL}{{\sc dpglasso}}
\newcommand{\GN}{{\sc gelnet}}
\newcommand{\DPGN}{{\sc dpgelnet}}
\newcommand{\RP}{{\sc rope}}
\newcommand{\Bo}{\boldsymbol}
\newcommand{\Ma}{\mathbf}

\newcommand{\oursoftware}{\textsf{GLassoElnetFast}}

\newcommand{\Rsoftware}{\textsf{R}}

\newtheorem{theorem}{Theorem}
\newtheorem{lemma}{Lemma}

\newtheorem{proof}{Proof}

\newtheorem{notation}{Notation}

\begin{document}
\title{Graphical Elastic Net and Target Matrices: Fast Algorithms and Software for Sparse Precision Matrix Estimation}

\author{Solt Kov\'acs${}^{1}$, Tobias Ruckstuhl${}^{1}$, Helena Obrist${}^{1}$, Peter B\"uhlmann${}^{1}$\\
\vspace{0.1cm}\\
{\small${}^{1}$Seminar for Statistics, ETH Zurich, Switzerland}}

\maketitle

\begin{abstract}
\label{chp:abstract}
We consider estimation of undirected Gaussian graphical models and inverse covariances in high-dimensional scenarios by penalizing the corresponding precision matrix. While single $L_1$ (Graphical Lasso) and $L_2$ (Graphical Ridge) penalties for the precision matrix have already been studied, we propose the combination of both, yielding an Elastic Net type penalty. We enable additional flexibility by allowing to include diagonal target matrices for the precision matrix. We generalize existing algorithms for the Graphical Lasso and provide corresponding software with an efficient implementation to facilitate usage for practitioners. Our software borrows computationally favorable parts from a number of existing packages for the Graphical Lasso, leading to an overall fast(er) implementation and at the same time yielding also much more methodological flexibility. 
\end{abstract}

\oursoftware\ R-package;
Graphical Lasso;
High-correlation;
High-dimensional;
ROPE;
Sparsity.

\section{Introduction and motivation} \label{chp:introduction}

The estimation of precision matrices (the inverse of covariance matrices) in high-dimensional settings, where traditional estimators perform poorly or do not even exist (e.g.~the inverse of the sample covariance matrix), has developed considerably over the past years. For multivariate Gaussian observations (with mean vector $\mu$ and covariance matrix $\mathbf{\Sigma}$) a zero off-diagonal entry of the precision matrix $\mathbf{\Sigma}^{-1}$ encodes conditional independence of the two corresponding variables given all the other ones. The resulting conditional independence graph, with variables as nodes and edges for nonzero off-diagonal entries, is called a Gaussian graphical model (GGM, \citealp{lauritzen}).

\citet{MeinsBuhl} proposed a nodewise regression approach (selecting one variable as the response while taking all remaining ones as predictors and repeating this for each variable) using the Lasso \citep{lasso_original}. The zero entries in the estimated Lasso regression coefficients serve as an estimate of the graphical model, i.e.~the zero pattern, but it does not give a full estimate of the precision matrix. To overcome this problem, \citet{MLErefit} proposed to use an unpenalized maximum likelihood estimator for the covariance matrix based on the previously estimated graph. \cite{Yuanlinprog} used the Dantzig selector \citep{CandesTao} for regression instead of the Lasso to get a preliminary estimate, followed by solving a linear program to obtain a symmetric matrix which is close to the preliminary estimate (even being positive definite with high probability).

A different approach for estimating precision matrices was proposed by \citet{YuanLin} based on the maximization of an $L_1$-penalized Gaussian log-likelihood for positive definite matrices. The $L_1$-penalty for the precision matrix encourages sparsity and thus the zero pattern of the estimated precision matrix serves as a direct estimate of the underlying GGM. Early proposals for the (challenging) computation of the estimates have been given by \citet{YuanLin}, \citet{Banerjee_dAspremont} as well as using the Graphical Lasso (\GL) by \citet{glasso}. The \GL\ algorithm has been widely used in the past and thus nowadays the $L_1$-penalized Gaussian log-likelihood precision matrix estimation problem itself is also called the Graphical Lasso (or \GL). Other recent approaches for the computation of the Graphical Lasso estimator include for example those of \cite{Scheinberg_GGM_computation, BIGQUIC_Hsieh, QUIC_Hsieh} and \cite{Atchade_Mazumder_Chen}. Theoretical properties of the Graphical Lasso estimator were investigated by \citet{YuanLin,Rothman,sparsistency} and \cite{Ravikumar} (up to a small difference of whether to penalize the diagonal elements of the precision matrix).

The Graphical Lasso has been modified and adapted to several more specific scenarios, e.g. in the presence of missing data \citep{missglasso}, change points \citep{Londschien}, unknown block structure \citep{Marlin}, in joint estimation proposals in the presence of groups sharing some information \citep{Guo_joint, Daneher_joint, Shan_unbalanced_multiclass}, with a prior clustering of the variables \citep{cluster_glasso}, or even for confidence intervals based on a desparsified Graphical Lasso estimator \citep{jankova}. While $L_1$-penalized estimation for GGMs is very popular, numerous other approaches have been proposed as well, e.g.~SCAD type penalty \citep{fan2009_chain_network}, the CLIME estimator \citep{CLIME_Cai_etal}, the SPACE estimator for partial correlation estimation \citep{Peng_BIC}, banded estimation \citep{banded} or shrinkage-type approaches \citep{wolf_math}.

\subsection{Related work} \label{sec:relatedwork}
Let $X_i \in \mathbb{R}^p$ for $i=1, \ldots, n $ be i.i.d. $p$-dimensional Gaussian random vectors with mean $\mu \in \mathbb{R}^p$ and positive definite covariance matrix $\mathbf{\Sigma} \in \mathbb{R}^{p \times p}$ with corresponding precision matrix $\mathbf{\Theta} = \mathbf{\Sigma}^{-1}$. Let $\mathbf{X} \in \mathbb{R}^{n \times p}$ be the data matrix comprising of the $n$ observations, i.e. $\mathbf{X}=(X_1,\ldots,X_n)^\top$. Without loss of generality assume that the columns of $\mathbf{X}$ are centered, such that the sample covariance matrix is given by $\mathbf{S}=\mathbf{X}^T\mathbf{X}/n$. Furthermore, for a given matrix $\mathbf{A} \in \mathbb{R}^{p \times p}$ let the matrix norms $\norm{\cdot}_1$ and $\norm{\cdot}_2$ be defined as $\norm{\mathbf{A}}_1=\sum_{i,j=1}^{p} \abs{A_{ij}} \text{ and } \norm{\mathbf{A}}_2=\norm{\mathbf{A}}_{F}=\sqrt{ \sum_{i,j=1}^{p} \abs{A_{ij}}^2} = \sqrt{ \text{tr}(\mathbf{A}^T \mathbf{A}) }$ and similarly let $\norm{\cdot}_{1^-}$ and $\norm{\cdot}_{2^-}$ be the norms without the diagonal entries, i.e. $\norm{\mathbf{A}}_{1^-}=\sum_{i\neq j} \abs{A_{ij}}$ and $\norm{\mathbf{A}}_{2^-}=\sqrt{ \sum_{i\neq j} \abs{A_{ij}}^2}$. Let diag$(\mathbf{A})$ be the diagonal matrix with the diagonal elements of the matrix $\mathbf{A}$ as its entries.
Finally, denote by $\textbf{A} \succ 0$ that the matrix $\textbf{A}$ is positive definite. Consider the following type of estimation problem for estimating the unknown precision matrix $\mathbf{\Theta}$ based on the $n$ observations:
\begin{equation}
\label{problem_setup}
\hat{\mathbf{\Theta}}(\lambda, \alpha, \textbf{T}) = 
    \underset{\mathbf{\Theta} \succ 0 }{\text{argmin}}  \{ -\text{log det} \mathbf{\Theta} + \text{tr} (\textbf{S} \mathbf{\Theta}) + \lambda (\alpha \norm{\mathbf{\Theta}- \textbf{T}}_1+\tfrac{1-\alpha}{2} \norm{\mathbf{\Theta}-\textbf{T}}_2^2) \},
\end{equation}
where $\textbf{T}\in \mathbb{R}^{p \times p}$ is a known (typically diagonal) positive semi-definite target matrix, $\lambda \geq 0$ and $\alpha \in [0,1]$ two tuning parameters. The problem can also be posed without penalizing the diagonal elements, i.e., with $\norm{\mathbf{\Theta}- \textbf{T}}_{1^-}$ and $\norm{\mathbf{\Theta}-\textbf{T}}_{2^-}^2$ and instead of the scalar $\lambda$ one could further generalize the estimation by allowing entry-wise penalties. For the ease of reading, we will focus on the more restrictive formulation in equation~\eqref{problem_setup} throughout the paper, but we enable the option for entry-wise penalties in our software. Note that due to the $L_2$-penalty term, the scaling of the variables matters. Thus, it matters whether as input $\textbf{S}$, the sample covariance or sample correlation matrix is provided. The letter would be recommended in practice in most scenarios.

The formulation \eqref{problem_setup} encompasses several previously proposed high-dimensional precision matrix estimators. The most classical one is the Graphical Lasso estimator \citep{glasso}, occurring in the case of $\alpha = 1$ and for the target matrix being the zero matrix. The $L_2$-penalty (i.e., $\alpha=0$) was recently proposed independently by \cite{Wieringen_Peeters} and \cite{Kuismin_ROPE}. The latter authors refer to it as the Ridge type operator for precision matrix estimation (\RP) with different target matrices, while \cite{Wieringen_Peeters} called it the Alternative Ridge Precision estimator, with Type I for the case of zero as target matrix and Type II for nonzero positive definite target matrix. For the $L_2$-penalized estimators, closed form solutions exist (even with target matrices), unlike for the Graphical Lasso, where various iterative procedures for the computations have been proposed, as mentioned earlier. The only proposal we are aware of incorporating target matrices for the $L_1$-penalty is that of \cite{Wieringen_iterative}. He is iteratively applying his generalized Ridge estimator (with elementwise differing~$\lambda$ values) to approximate the loss function for the $L_1$ case. However, when aiming for reasonably accurate estimates, this approach is computationally not attractive even for only moderately sized problems.
Now let us mention the combination of $L_1$ and $L_2$-penalties, i.e., the case $\alpha \in (0,1)$ for problem~\eqref{problem_setup}. Analogously to the Elastic Net regression \citep{elastic_net}, we call this version the Graphical Elastic Net. Adapting the iterative procedure of \cite{Wieringen_iterative} for the Elastic Net problem is possible, but again, is not attractive computationally (see section~\ref{chp:comptime}). Other work in the context of Elastic Net we are aware of are without target matrices (i.e. $\textbf{T}=\textbf{0}$). Genevera Allen in her 2010 Stanford University PhD thesis \citep[Algorithm 2]{allen_phd} adapted the \GL\ algorithm of \cite{glasso}, similar to
our proposal in section \ref{sec:gelnetdpgelnet}. %
\cite{Atchade_Mazumder_Chen} proposed stochastic proximal optimization methods to obtain near-optimal (i.e., approximate) solutions for regularized precision matrix estimation, and their approach incorporates also Elastic Net penalties (without target). They focus on large-scale problems, where the computation of exact solutions becomes impractical. 
Lastly, \cite{rothman_elastic} considers a sparse covariance matrix estimator and his proposed algorithm bears resemblance to our algorithms for Elastic Net type penalties.
We are not aware of publicly available software corresponding to these proposals with Elastic Net penalties, moreover,  efficient software for target matrices is limited to the Ridge penalty.

\subsection{Our contribution and outline} \label{sec:ourcontribution}
We focus on the estimation problem \eqref{problem_setup}. Elastic Net type penalties (for Gaussian log-likelihood based precision matrix estimation) could help to obtain both the advantages of $L_2$-penalized estimation (as presented recently by \cite{Wieringen_Peeters} and \cite{Kuismin_ROPE}) as well as benefits of $L_1$-penalization such as sparse precision matrix estimates that are desired for graphical models. Moreover, the inclusion of suitable target matrices could considerably improve estimation results. However, as discussed previously, only special cases have been treated so far in the literature both in terms of  available algorithms as we well as corresponding software. Our goal is to contribute to computational approaches for the general problem \eqref{problem_setup}, i.e., to develop new algorithms in order to allow estimation in more flexible ways, and also to provide software with efficient implementations of these proposals for practitioners. A practice oriented introduction to targets, a motivating real data example and some code snippets showing the usage of the package in section~\ref{chp:methodology_and_software} are aimed to facilitate the understanding and usage.

Our algorithms are building upon the \GL\ algorithm of \cite{glasso} and the more recent \DPGL\ algorithm of \cite{Mazumder2012_dpglasso} that offers certain advantages over its predecessor. We describe these two algorithms (for the case without target matrices) in section \ref{sec:glassodpglasso}. We then introduce our modifications in section \ref{sec:gelnetdpgelnet} that lead to our \GN\ (Graphical Elastic Net) and \DPGN\ algorithms. In section~\ref{chp:target} we further generalize these approaches to include diagonal target matrices. While diagonal matrices are less flexible than arbitrary positive semi-definite target matrices, we note that in practice many fall into this category (e.g.~all target matrices that were used in simulations by \cite{Wieringen_Peeters} and by \cite{Kuismin_ROPE}). Our developed algorithms are supported by mathematical derivations and theory. We present simulations in section~\ref{chp:simulation} comparing our new estimators with Elastic Net penalties and target matrices with the plain Graphical Lasso and also recent Ridge-type approaches. Competitive computational performance is demonstrated in section~\ref{chp:comptime} by benchmarking our software to widely used packages such as the \textsf{glasso} \citep{glasso_package}, the \textsf{glassoFast} \citep{glassofast_package}, and the \textsf{dpglasso} \citep{dpglasso_package} \textsf{R}-packages in terms of computational times.

\section{Methodology and software} 
\label{chp:methodology_and_software}

\subsection{Target types} \label{sec:target_types}
We present some target types that can be used within equation~\eqref{problem_setup} and which were used in the simulations in section \ref{chp:simulation}. 

\begin{itemize}
    \item True Diagonal: Taking the diagonal of the underlying true precision matrix. Of course, using this target is only possible in simulation settings, where the original precision matrix is known.
    \item Identity: Taking the identity matrix as target. This choice of the target is very conservative (i.e., having rather small entries) when the input $\textbf{S}$ in problem~\eqref{problem_setup} is the empirical correlation matrix.
    \item $v$-Identity: Multiplying the identity matrix with a scalar $v$, where $v$ is the inverse of the mean of all diagonal entries in the sample covariance matrix (see \citealp{Kuismin_ROPE}).
    \item Eigenvalue: The ``DAIE'' target type from the \textsf{default.target} function of the \textsf{rags2ridges} \textsf{R}-package \citep{rags2ridges_package}. This takes a diagonal matrix with average of inverse nonzero eigenvalues of the sample covariance matrix as entries. Eigenvalues under a certain threshold are set to zero.
    \item Maximal Single Correlation: For variable $j$ take from the remaining variables the one that has the highest absolute correlation with $j$ and denote this as $k$. The corresponding correlation is denoted as $\rho_{jk}$. Then set $\mathbf{T}_{jj} = ((1-\abs{\rho_{jk}})\cdot S_{jj})^{-1}$, where $S_{jj}$ is the $j$-th diagonal element of the sample covariance matrix $S$.
    \item Nodewise regression
    \citep{MeinsBuhl}: The idea is similar as with the maximal single correlation, but we would like to allow for more than one predictor. For a variable $j$ apply a 10-fold cross-validation using Lasso regression (using all other variables than the $j$-th one as predictors). Then set the $j$-th diagonal entry of the target matrix as the inverse of the minimal cross-validated error variance. Under suitable conditions, the target matrix diagonals converge to the diagonal values of the true precision matrix $\Theta$.
\end{itemize}

Some further targets are implemented in the \textsf{default.target} function of the \textsf{rags2ridges} \textsf{R}-package of \cite{rags2ridges_package}.


\subsection{Motivation: a small real data example} \label{chp:realdata}

In order to illustrate differences between the methodologies, we apply \GL\ and \GN\ (with $\alpha = 0.5$) to a real data example from \cite{wille_realdata}. The object to study are expressions of $p=39$ isoprenoid genes from $n=118$ samples from the plant Arabidopsis Thaliana. As there is no solid ground truth available, we only briefly show that different estimation approaches lead to different results and thus the additional flexibility provided by target matrices and elastic net type penalties provides additional options for real data applications compared to the \GL. In all examples below the penalization parameter for each algorithms is chosen such that the resulting precision matrix only has 200 non-zero off-diagonal entries, which leads to 100 edges between the genes. In the first example (Figure~\ref{realdata1}) we take the sample correlation matrix for the genes and apply both \GL\ and \GN\ (with $\alpha = 0.5$) without target matrix (i.e.~$\textbf{T}=\textbf{0}$).
\begin{figure}[H]
\includegraphics[width=1\textwidth]{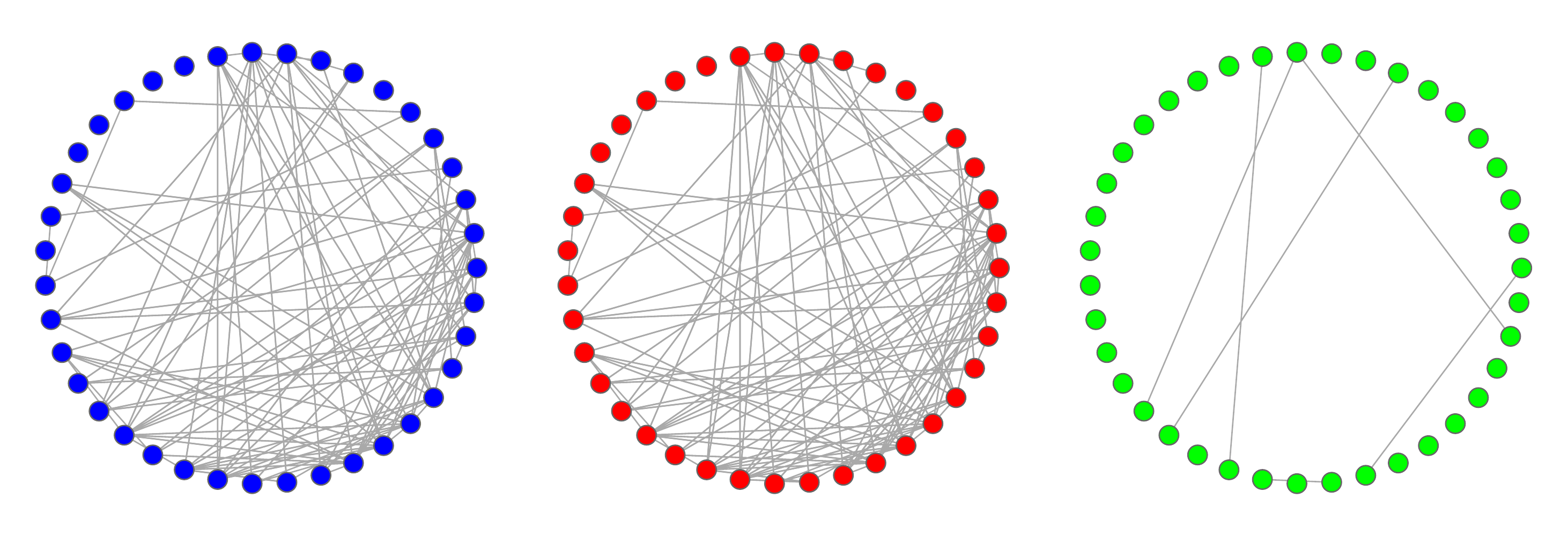}
  \caption{Original data and no target is used (i.e.~$\textbf{T}=\textbf{0}$). The 100 edges chosen by \GL\ (on the left) and \GN\  with $\alpha = 0.5$ (in the middle). Edges, which are only present in one algorithm are displayed on the right. 97 edges are in common.}
\label{realdata1}
\end{figure}
In the second example we adjust for the first principal component and then apply the same procedure to the new sample correlation matrix. Such an adjustment is common to remove potential hidden confounding variables. %
As shown on the green right plot of Figure~\ref{realdata2}, in this case fewer edges are in common.
\begin{figure}[H]
  \includegraphics[width=1\textwidth]{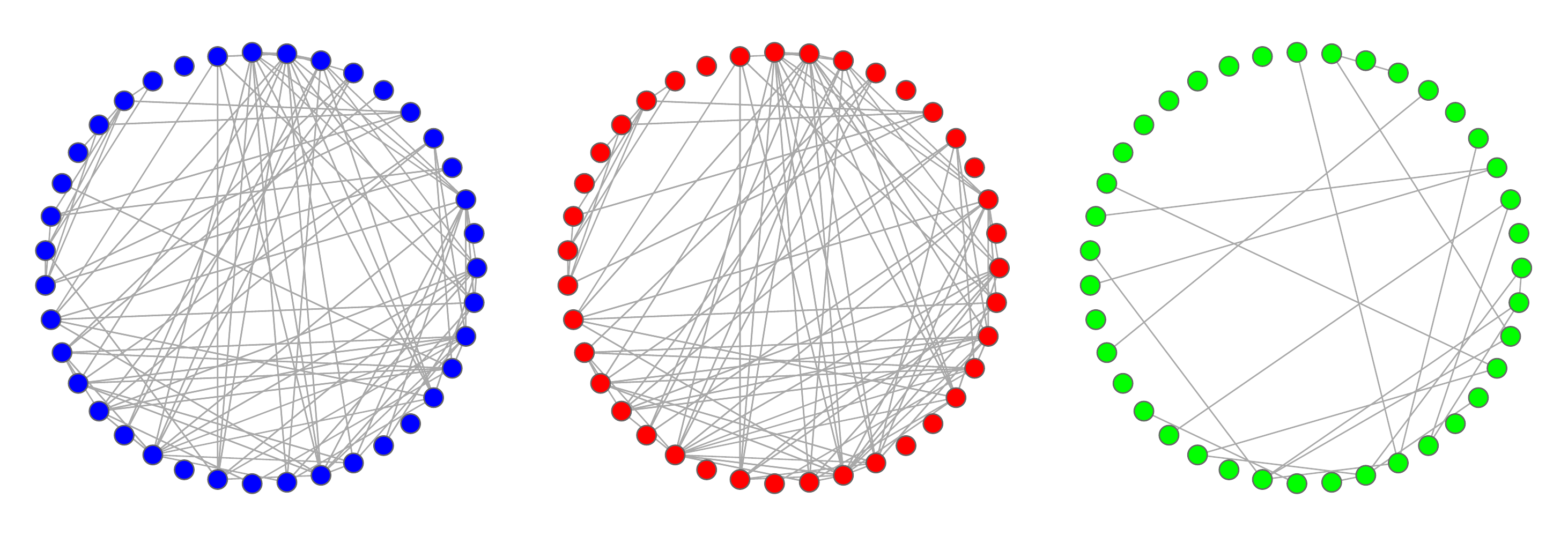}
  \caption{Data adjusted for the first principal component and no target is used. The 100 edges chosen by \GL\ (on the left) and \GN\ with $\alpha = 0.5$ (in the middle). Edges, which are only present in one algorithm are displayed on the right. 89 edges are in common. %
    \label{realdata2}
  }
\end{figure}
Instead of adjusting for the first principal component, we use in the third example the \textit{Maximal Single Correlation} target from subsection \ref{sec:target_types}. As shown on the green right plot of Figure~\ref{realdata3}, in this case even fewer edges are in common.
\begin{figure}[H]
  \includegraphics[width=1\textwidth]{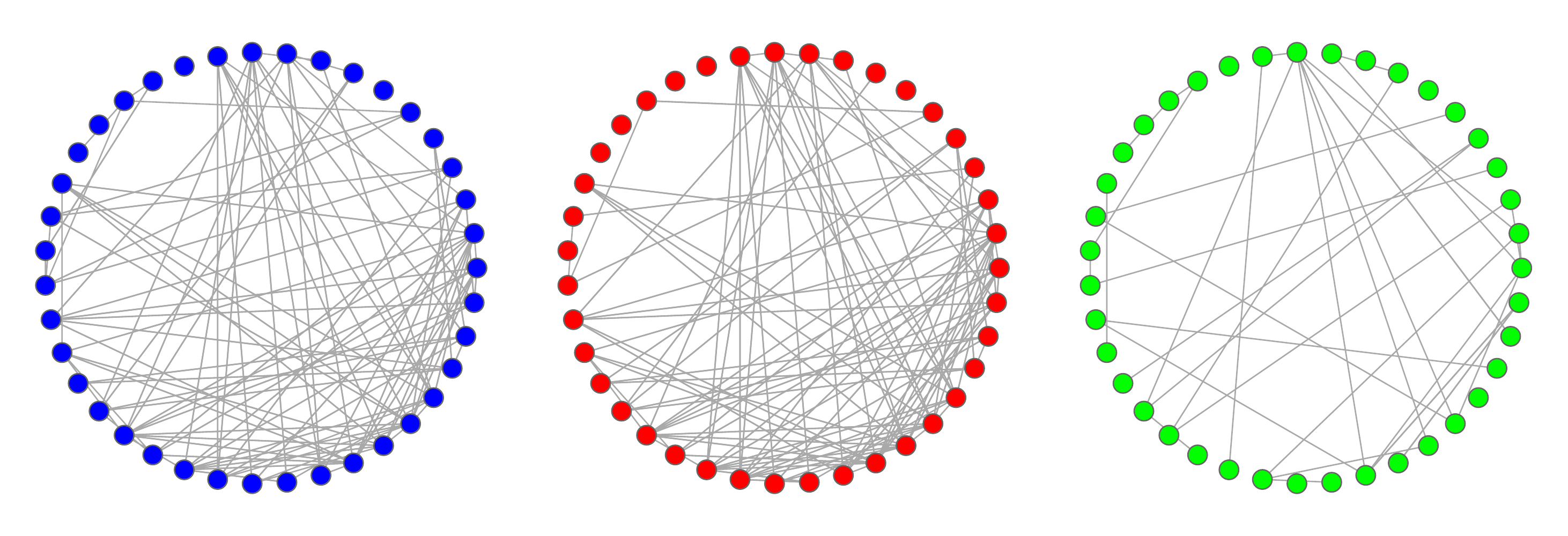}
  \caption{Original data and the \textit{Maximal Single Correlation} target is used. The 100 edges chosen by \GL\ (on the left) and \GN\  with $\alpha = 0.5$ (in the middle). Edges, which are only present in one algorithm are displayed on the right. 83 edges are in common.}
\label{realdata3}
\end{figure}


\subsection{Software: the \texttt{R}-package \oursoftware}
\label{chp:software}
In the following we provide a few examples on how to use our software.

\paragraph{Installing the \oursoftware\ \Rsoftware-package.} This can be done as follows.
\begin{lstlisting}[language=R, 
basicstyle=\small\ttfamily,
    stringstyle=\color{orange},
%otherkeywords={0,1,2,3,4,5,6,7,8,9},
    morekeywords={install_, github, data, gelnet, dpgelnet, glasso, dpglasso, rope,default.target, target, TRUE,FALSE, rags2ridges, ridgeP, crossvalidation},
    deletekeywords={frame,length,as,character, grid},
    keywordstyle=\color{blue},
    commentstyle=\color{olive},
]
# The package can be installed directly from github (using devtools package):
# install.packages(devtools)
devtools::install_github("TobiasRuckstuhl/GLassoElnetFast")
library(GLassoElnetFast)
# We use the first 5 columns of the FHT data from the gcdnet package and lambda = 0.1. 
library(gcdnet); data(FHT); X <- FHT$x[,1:5]; lambda <- 0.1
\end{lstlisting}

\paragraph{Estimation with no target matrices.} We show here the options for the Graphical Lasso, the Graphical Elastic Net and the Rope method (Ridge penalty). 

\begin{lstlisting}[language=R, 
basicstyle=\small\ttfamily,
    stringstyle=\color{orange},
%otherkeywords={0,1,2,3,4,5,6,7,8,9},
    morekeywords={install_, github, data, gelnet, dpgelnet, glasso, dpglasso, rope,default.target, target, TRUE,FALSE, rags2ridges, ridgeP, crossvalidation},
    deletekeywords={frame,length,as,character, grid},
    keywordstyle=\color{blue},
    commentstyle=\color{olive},
]
#######################   Example 1: zero as target matrix   #######################
# GRAPHICAL LASSO penalty (alpha = 1) for the input correlation matrix
fitGlasso  <- gelnet(S = cor(X), lambda = lambda, alpha = 1)   # gelnet algorithm
fitGlasso2 <- dpgelnet(S = cor(X), lambda = lambda, alpha = 1) # dpgelnet algorithm
fitGlasso3 <- glasso::glasso(s = cor(X), rho = lambda)         # classical glasso
fitGlasso4 <- dpglasso::dpglasso(Sigma = cor(X), rho = lambda) # dpglasso 
# up to differences due to stopping criterion, our gelnet and dpgelnet implementation
# give the same results for the precision matrix as the existing glasso and dpglasso:
fitGlasso$Theta; fitGlasso2$Theta; fitGlasso3$wi; fitGlasso4$X # all same

# ELASTIC NET penalty (alpha = 0.5) for the input correlation matrix
fitGelnet  <- gelnet(S = cor(X), lambda = lambda, alpha = 0.5)  # gelnet algorithm
fitGelnet2 <- dpgelnet(S = cor(X), lambda = lambda, alpha = 0.5)# dpgelnet algorithm
fitGelnet$Theta; fitGelnet2$Theta                               # all same

# RIDGE penalty (alpha = 0) for the input correlation matrix
fitROPE1 <- gelnet(S = cor(X), lambda = lambda, alpha = 0)  # ROPE via gelnet		
fitROPE2 <- rope(S = cor(X), lambda = lambda) 	            # ROPE via closed form	
fitROPE3 <- rags2ridges::ridgeP(cor(X), lambda = lambda, target =
                    matrix(0, ncol(X), ncol(X))) # from rags2ridges package
fitROPE1$Theta; fitROPE2; fitROPE3               # all same
\end{lstlisting}

\paragraph{Estimation with the identity target matrix.} We show here the options for the Graphical Lasso, the Graphical Elastic Net and the Rope method (Ridge penalty).

\begin{lstlisting}[language=R, 
basicstyle=\small\ttfamily,
    stringstyle=\color{orange},
%otherkeywords={0,1,2,3,4,5,6,7,8,9},
    morekeywords={install_, github, data, gelnet, dpgelnet, glasso, dpglasso, rope,default.target, target, TRUE,FALSE, rags2ridges, ridgeP, crossvalidation},
    deletekeywords={frame,length,as,character, grid},
    keywordstyle=\color{blue},
    commentstyle=\color{olive},
]

#####################   Example 2: Identity as target matrix    #####################
# GRAPHICAL LASSO penalty (alpha = 1), correlation input, Identity target
fitGlassoId <- gelnet(S = cor(X), lambda = lambda, alpha = 1, 
                            Target = target(Y = X, type = "Identity", cor = T))
# ELASTIC NET penalty (alpha = 0.5), correlation input, Identity target
fitGelnetId <- gelnet(S = cor(X), lambda = lambda, alpha = 0.5, 
                            Target = target(Y = X, type = "Identity", cor = T))
# RIDGE penalty (alpha = 1), correlation input, Identity target; following 3 options 
fitROPEId1  <- gelnet(S = cor(X), lambda = lambda, alpha = 0,
        Target = target(Y = X, type = "Identity", cor = T)) # via gelnet algorithm
fitROPEId2  <- rope(S = cor(X), lambda = lambda, 
        Target = target(Y = X, type = "Identity", cor = T)) # closed form
fitROPEId3  <- rags2ridges::ridgeP(cor(X), lambda = lambda, target =    
        rags2ridges::default.target(S = cor(X), type = "DEPV")) # another closed form
fitGlassoId$Theta; fitGelnetId$Theta; fitROPEId1$Theta; fitROPEId2  # compare fits

\end{lstlisting}

\paragraph{Estimation with the identity target matrix and using cross-validation.} We show here the options for the Graphical Elastic Net.

\begin{lstlisting}[language=R, 
basicstyle=\small\ttfamily,
    stringstyle=\color{orange},
%otherkeywords={0,1,2,3,4,5,6,7,8,9},
    morekeywords={install_, github, data, gelnet, dpgelnet, glasso, dpglasso, rope,default.target, target, TRUE,FALSE, rags2ridges, ridgeP, crossvalidation},
    deletekeywords={frame,length,as,character, grid},
    keywordstyle=\color{blue},
    commentstyle=\color{olive},
]
###############   Example 3: cross-validation with Identity target   ###############
lambda_grid     <- 0.9^c(0:40)
fitGelnetIdCV	<- crossvalidation(Y = X, lambda = lambda_grid, alpha = 0.5, 
                                        cor = T, type = "Identity")
# the optimal lambda obtained and the corresponding fit for the precision matrix is:
fitGelnetIdCV$optimal; fitGelnetIdCV$Theta
gelnet(S = cor(X), lambda = fitGelnetIdCV$optimal, alpha = 0.5, Target = 
    target(Y = X, type = "Identity", cor = T))$Theta  # same fit with optimal lambda
    
    
\end{lstlisting}

\paragraph{Further remarks about the \textsf{gelnet} and \textsf{dpgelnet} functions.} They include the following: 
\begin{itemize}
\item Instead of a scalar $\lambda$, one can also provide entry-wise penalties via the argument \textsf{lambda} (which needs to be a vector or a matrix in this case).
\item Besides the estimated precision matrix (\textsf{Theta}), also an estimate of the covariance matrix (\textsf{W}) is returned. This can be accessed by e.g.~\textsf{fitGelnet\$W}.
\item One can choose whether to penalize the diagonal via the \textsf{penalize.diagonal} argument. %
\item One can provide warm starts (via arguments \textsf{Theta} and \textsf{W}).
\item \textsf{niter}, \textsf{del} and \textsf{conv} are also part of the output yielding information about the number of iterations, change in parameter value and convergence (\textsf{TRUE} or \textsf{FALSE}) with the set number of iterations and thresholds. One can set various thresholds (arguments \textsf{outer.thr} and \textsf{inner.thr}) and iteration numbers (arguments \textsf{outer.maxit} and \textsf{inner.maxit}) when aiming for more accurate results or in case algorithms face convergence issues.
\end{itemize}

\noindent Other notable arguments for the \textsf{gelnet} function: 
\begin{itemize}
\item The argument \textsf{zero} allows to specify indices of the precision matrix which are constrained to be zero.
\item The argument \textsf{Target} allows to specify a diagonal target matrix. Target matrices mentioned in this paper are implemented via the function \textsf{target} and some further ones are implemented in the \textsf{default.target} function of the \textsf{rags2ridges} \textsf{R}-package.
\end{itemize}

\noindent Further explanations can be obtained in the following help files:
\begin{lstlisting}[language=R, 
basicstyle=\small\ttfamily,
    stringstyle=\color{orange},
%otherkeywords={0,1,2,3,4,5,6,7,8,9},
%    morekeywords={install_, github, data, gelnet, dpgelnet, glasso, dpglasso, rope,default.target, target, TRUE,FALSE, rags2ridges, ridgeP},
    deletekeywords={ data,frame,length,as,character},
    keywordstyle=\color{blue},
    commentstyle=\color{olive},
]
help(gelnet)		# gelnet algorithm with option for target matrices
help(dpgelnet)          # dpgelnet algorithm
help(target)		# implemented target matrices described in this paper
help(rope)              # closed form solution for ROPE estimator 
help(crossvalidation)	# aiding CV to find optimal tuning parameter
\end{lstlisting}

Our implementations in the \oursoftware\ \Rsoftware-package build upon the \textsf{glasso} \textsf{R}-package \citep{glasso_package} and the \textsf{dpglasso} \textsf{R}-package \cite{dpglasso_package} by suitably modifying them. Additionally, we also utilized a faster implementation from the
\textsf{glassoFast} \Rsoftware-package \citep{glassofast_package, glassofast_paper}. Moreover, whenever possible, we even improve on the original versions. For example, to achieve further speed-ups, we incorporate block diagonal screening \citep{Mazumder2012_connected_components, Witten2011_connected_components} which was missing in the \textsf{dpglasso} and \textsf{glassoFast} packages. Compared to the \textsf{dpglasso} package, our implementation relies even more on \textsf{Fortran}, which is beneficial for its speed, while for the estimation we also allow the $\norm{\cdot}_{1^-}$ penalty (instead of $\norm{\cdot}_1$ only) and entry-wise penalties. 
Besides such technical improvements on existing software (which lead to some computational speed-up as shown in section~\ref{chp:comptime}), the main new features are the additional flexibility of allowing Elastic Net penalties and diagonal target matrices as described earlier. Our implementation is planned to be made available as an \textsf{R}-package on \textsf{CRAN}. 

\section{From \GL\ and \DPGL\ to \GN\ and \DPGN} 
\label{chp:exist modif}
\subsection{\GL\ and \DPGL} \label{sec:glassodpglasso}
In this section we briefly present the \GL\ algorithm by \citet{glasso} as well as the \DPGL\ algorithm by \cite{Mazumder2012_dpglasso} and set up some notation which we will rely on when presenting in the next subsection~\ref{sec:gelnetdpgelnet} the modified versions with the Elastic Net penalties called \GN\ and \DPGN. We closely follow the derivation presented in \citet{Mazumder2012_dpglasso}. Both the \GL\ and the \DPGL\ algorithm seek to minimize the $L_1$-regularized negative log-likelihood over all positive definite precision matrices~$\Ma{\Theta}$:
\begin{equation*}
    \hat{\mathbf{\Theta}}(\lambda, \alpha=1, \textbf{T}=\mathbf{0}) = 
    \underset{\mathbf{\Theta} \succ 0 }{\text{argmin}} \{ -\text{log det} \Ma{\Theta} + \text{tr} (\textbf{S} \Ma{\Theta}) + \lambda \norm{\Ma{\Theta}}_1 \}.
\end{equation*}
We call the function inside the argmin the \textit{Graphical Lasso loss}, which is a special case of \eqref{problem_setup}. We assume the tuning parameter $\lambda \geq 0$ to be a scalar (but as further generalization, element-wise penalties encoded in a matrix are also possible). The positive semi-definite and symmetric matrix \textbf{S} required as the input is usually the empirical covariance or correlation matrix. The algorithms work with the normal equations corresponding to the Graphical Lasso loss:
\begin{equation}\label{Normal: Glasso}
    -\Ma{\Theta}^{-1} +  \textbf{S}  + \lambda  \Ma{\Gamma} =\textbf{0}, \quad  -\textbf{W} +  \textbf{S}  + \lambda  \Ma{\Gamma}=\textbf{0},
\end{equation}
where $\textbf{W} \coloneqq \Ma{\Theta}^{-1}$ is called the working covariance matrix and the matrix $\Ma{\Gamma}$ denotes the component-wise signs of $\Ma{\Theta}$ with $\Ma{\Gamma}_{i,j} \in [-1,1]$ for $\Ma{\Theta}_{i,j}=0$ . The approach used for solving the normal equations is to update one dimension (row and column) at a time while leaving the remaining ones fixed. The update requires solving a suitable penalized regression problem which can be performed efficiently by applying coordinate descent \citep{coordinatewise}. After updating, one proceeds with the next dimension (row and column), until all dimensions have been updated once. Typically one repeats these update cycles for all the variables as long as differences are above a certain threshold used as a stopping criterion.
\begin{notation}
Partition $\textbf{S}, \Ma{\Theta}, \textbf{W} \text{ and } \Ma{\Gamma}$ into block form  as follows: a matrix with dimensions $(p-1) \times (p-1)$ (e.g. $\textbf{S}_{11}$), two vectors of length $(p-1)$ (e.g. $\textbf{s}_{12}$) and a scalar (e.g. $s_{22}$):
\begin{equation}\label{Block}
    \textbf{S}= \begin{pmatrix}
    \textbf{S}_{11} & \textbf{s}_{12} \\
    \textbf{s}_{21} & s_{22} 
  \end{pmatrix},
  \Ma{\Theta}= \begin{pmatrix}
    \Ma{\Theta}_{11} & \Bo{\theta}_{12} \\
    \Bo{\theta}_{21} & \theta_{22} 
  \end{pmatrix},
  \textbf{W}= \begin{pmatrix}
    \textbf{W}_{11} & \textbf{w}_{12} \\
    \textbf{w}_{21} & w_{22} 
  \end{pmatrix},
  \Ma{\Gamma}= \begin{pmatrix}
    \Ma{\Gamma}_{11} & \Bo{\gamma}_{12} \\
    \Bo{\gamma}_{21} & \gamma_{22}
  \end{pmatrix}.
\end{equation}
\end{notation} 
The main difference of \GL\ and \DPGL\ lies in the fact that \GL\ is actively working with $\textbf{W}$, while \DPGL\ works with its inverse $\Ma{\Theta}$. Furthermore, \GL\ and \DPGL\ use different block-matrix representations for $\textbf{W}$, where properties of inverses of block-partitioned matrices are used and $\Ma{\Theta}=\Ma{W}^{-1}$:
\begin{align}
    \Ma{W}_{Glasso} &=
  \begin{pmatrix}
    \textbf{W}_{11} & \textbf{w}_{12} \\
    \textbf{w}_{21} & w_{22} 
  \end{pmatrix}
  =
  \begin{pmatrix}
    (\Ma{\Theta}_{11}-\frac{\Bo{\theta}_{12}\Bo{\theta}_{21}}{\theta_{22}})^{-1} &
    - \Ma{W}_{11} \frac{\Bo{\theta}_{12}}{\theta_{22}}\\
    . &
    \frac{1}{\theta_{22}}+\frac{\Bo{\theta}_{21}\Ma{W}_{11}\Bo{\theta}_{12}}{\theta^2_{22}}
  \end{pmatrix} \label{WGlasso} ,\\
    \Ma{W}_{DPGlasso} &=
    \begin{pmatrix}
    \textbf{W}_{11} & \textbf{w}_{12} \\
    \textbf{w}_{21} & w_{22} 
     \end{pmatrix}
    =
    \begin{pmatrix}
    \Ma{\Theta}^{-1}_{11}-\frac{\Ma{\Theta}^{-1}_{11}\Bo{\theta}_{12}\Bo{\theta}_{21}\Ma{\Theta}^{-1}_{11}}{\theta_{22}-\Bo{\theta}_{21}\Ma{\Theta}^{-1}_{11}\Bo{\theta}_{12}} &
    - \frac{\Ma{\Theta}^{-1}_{11}\Bo{\theta}_{12}}{\theta_{22}-\Bo{\theta}_{21}\Ma{\Theta}^{-1}_{11}\Bo{\theta}_{12}}\\
    . &
    \frac{1}{\theta_{22}-\Bo{\theta}_{21}\Ma{\Theta}^{-1}_{11}\Bo{\theta}_{12}}
  \end{pmatrix} \label{WDPGlasso}.
\end{align}
As the remaining steps of the derivation (e.g. the regression problems involved in the row and column updates as well as the coordinate descent procedure to obtain solutions for it) are special cases of what is discussed in the next subsection~\ref{sec:gelnetdpgelnet}, we only present the final algorithms here. Note that in Algorithm~\ref{alg:algorithm1} (and Algorithm~\ref{alg:algorithm2}) \textbf{I} denotes the identity matrix (of dimension $p\times p$).

\begin{algorithm}[H]
    \caption{\GL\ algorithm}
  \begin{algorithmic}[1]
    \State Initialize $\textbf{W}_{\text{init}}= \textbf{S} + \lambda\textbf{I}$.
    \State Cycle around the columns repeatedly, performing the following steps till convergence:
    \begin{algsubstates}
        \State Rearrange the rows/columns so that the currently updated column is last (implicitly).
        \State Solve the following Lasso problem with coordinate descent to get $\hat{\Bo{\beta}}$:
        \begin{equation*}
            \hat{\Bo{\beta}}= \underset{\Bo{\beta} \in \mathbb{R}^{p-1}}{\text{argmin }} \{ \tfrac{1}{2} \Bo{\beta}^T \textbf{W}_{11} \Bo{\beta} + \Bo{\beta}^T \textbf{s}_{12} + \lambda \norm{\Bo{\beta}}_1 \}  .
        \end{equation*}
        As warm start for $\Bo{\beta}$ use the solution from the previous round for this row/column.
        \State Update the off-diagonal of the working covariance matrix as  $\hat{\Ma{w}}_{12} = \textbf{W}_{11} \hat{\Bo{\beta}}$ (and similarly for $\hat{\Ma{w}}_{21}$), but do not change the diagonal entry $w_{22}$.
        \State Save $\hat{\Bo{\beta}}$ for this row/column in a matrix $\Ma{B}$. 
    \end{algsubstates}
    \State Finally, for every row/column, compute the diagonal entries of $\Ma{\Theta}$ using $\hat{\theta}_{22} = \frac{1}{w_{22}-\hat{\Bo{\beta}} \hat{\Ma{w}}_{12} }$ and obtain the off-diagonal entries of $\Ma{\Theta}$ from the matrix $\Ma{B}$, where $\hat{\Bo{\theta}}_{12}=  - \hat{\theta}_{22} \hat{\Bo{\beta}}$ and $\hat{\Bo{\theta}}_{21}= \hat{\Bo{\theta}}_{12}^T$.
  \end{algorithmic} 
  \label{alg:algorithm1}
\end{algorithm}
\begin{algorithm}[H]
  \caption{\DPGL\ algorithm}
  \begin{algorithmic}[1]
    \State Initialize $\Ma{\Theta}_{\text{init}}=\text{diag}(\textbf{S}+ \lambda \textbf{I})^{-1}$ and $\textbf{W}_{\text{init}}= \textbf{S} + \lambda\textbf{I}$.
    \State Cycle around the columns repeatedly, performing the following steps till convergence:
    \begin{algsubstates}
        \State Rearrange the rows/columns so that the currently updated column is last (implicitly).
        \State Solve the following quadratic program with coordinate descent for $\Bo{\gamma} \in \mathbb{R}^{p-1}$:
        \begin{equation*}
           \hat{\Bo{\gamma}} = \underset{  \norm{\Bo{\gamma}}_{\infty} \leq \lambda   }{\text{argmin }} \{ \tfrac{1}{2} (\textbf{s}_{12}+ \Bo{\gamma})^T \Ma{\Theta}_{11} (\textbf{s}_{12}+ \Bo{\gamma}) \} .
        \end{equation*}
        \State Update $\hat{\Bo{\theta}}_{12} = \tfrac{- \Ma{\Theta}_{11} (\textbf{s}_{12}+ \hat{\Bo{\gamma}}) }{w_{22}}$ and $\hat{\Bo{\theta}}_{21} = \hat{\Bo{\theta}}_{12}^T$.
        \State Update $\hat{\theta}_{22} = \tfrac{1- (\textbf{s}_{12}+\Ma{\hat{\gamma}})^T \hat{\Bo{\theta}}_{12}}{w_{22}}$.
        \State Update $\Ma{\hat{w}}_{12}  = \textbf{s}_{12}+ {\hat{\Bo{\gamma}}}$ and $\Ma{\hat{w}}_{21} = \Ma{\hat{w}}_{12}^T$, but do not change the diagonal entry $w_{22}$.
    \end{algsubstates}
  \end{algorithmic} 
  \label{alg:algorithm2}
\end{algorithm}
\vspace{1cm}
If the primary interest lies in the estimation of $\Ma{\Theta}$, the \DPGL\ algorithm offers several advantages. It yields both a sparse and positive definite $\Ma{\Theta}$, while \GL\ only ensures a sparse solution, which is not necessarily positive definite in all cycles of the algorithm and thus rarely also when stopping. Moreover, \DPGL\ converges with all positive definite warm starts $\Ma{\Theta}_{w}$. \GL\ is guaranteed to maintain positive definite updates in each step only if the warm start $\Ma{W}_w$ fulfills the conditions $\Ma{W}_w \succ 0$ and $\norm{\Ma{W}_w - \Ma{S}}_{\infty} \leq \lambda$ (which is the case for example for the standard initialization $\textbf{W}_{\text{init}}= \textbf{S} + \lambda\textbf{I}$). For detailed results and differences see \cite{Mazumder2012_dpglasso}.

\subsection{\GN\ and \DPGN} \label{sec:gelnetdpgelnet}
In this section the modified versions of the \GL\ and \DPGL\ algorithms, named \GN\ and \DPGN\ are derived. The specific implementations for \GN\ and \DPGN\ utilize the shell of the \GL\ code provided by \cite{glasso_package} and the \DPGL\ code provided by \cite{dpglasso_package}, see section~\ref{chp:software}. Penalizing the negative log-likelihood with a combination of $L_1$ and $L_2$ terms leads to the following special case of problem \eqref{problem_setup}:
\begin{equation}\label{Elastic Net}
    \hat{\mathbf{\Theta}}(\lambda, \alpha, \textbf{T}=\mathbf{0}) = 
    \underset{\mathbf{\Theta} \succ 0 }{\text{argmin}}  \{ -\text{log det} \Ma{\Theta} + \text{tr} (\textbf{S} \Ma{\Theta}) + \lambda (\alpha \norm{\Ma{\Theta}}_1+\tfrac{1-\alpha}{2} \norm{\Ma{\Theta}}_2^2)        \} ,
\end{equation}
where $\lambda$ is a non-negative tuning parameter, $\alpha \in [0,1]$ is another tuning parameter, %
and $\Ma{S}$ is a positive semi-definite matrix (typically the covariance or correlation matrix). %
Minimizing expression \eqref{Elastic Net} is called the \textit{Graphical Elastic Net} or in short the \GN\ problem. The corresponding normal equations are
\begin{equation}\label{Norm: Elastic Net}
    -\Ma{\Theta}^{-1} +  \textbf{S}  + \lambda \alpha \Ma{\Gamma} + \lambda (1-\alpha) \Ma{\Theta}=\textbf{0}, \quad -\Ma{W} +  \textbf{S}  + \lambda \alpha \Ma{\Gamma} + \lambda (1-\alpha) \Ma{\Theta}=\textbf{0} ,
\end{equation}
where $\Ma{\Gamma}$ is the sign matrix of $\Ma{\Theta}$ and $\Ma{W}$ is the working covariance matrix similar to previous notation from \eqref{Normal: Glasso}. The \GN\ and \DPGN\ algorithms follow the block coordinate descent approach of the \GL\ and \DPGL\ algorithms in solving the normal equations \eqref{Norm: Elastic Net}. Namely, update one row/column at the time while leaving the remaining ones fixed. The update requires solving another optimization problem (for a regression problem), which can be performed efficiently by applying coordinate descent \citep{coordinatewise}. After updating, one proceeds with the next row/column update, until all of them have been updated once. Typically one repeats these update cycles for all the variables as long as differences are above the threshold used as the stopping criterion. Both algorithms work actively with $\Ma{\Theta}$ and $\Ma{W}$ simultaneously, in contrast to \GL\ and \DPGL. The detailed derivation of \GN\ and \DPGN\ can be split into two problems:
\begin{itemize}
    \item Problem 1) Derive the formula for the row/column updates as well as the therein involved optimization problem.
    \item Problem 2) Apply coordinate descent for the optimization from Problem 1).
\end{itemize}
Coordinate descent \citep{Tseng_coordinate, coordinatewise} performs componentwise updates leaving the other coordinates fixed. If the function to be minimized is convex and can be decomposed into a differentiable, convex function plus a sum of convex functions of each individual parameter, then iteratively updating each coordinate is guaranteed to converge to the global minimizer. For more details, see \citet{Tseng_coordinate, coordinatewise} or \cite{SLS_Hastie_etal}. The functions to be minimized at each of the row/column updates of \GN\ and \DPGN\ fulfill this property and hence only the componentwise updates have to be derived.

\subsubsection{\GN}

\paragraph{Solution to Problem 1)} 
Consider the $p$-th row of the normal equations in \eqref{Norm: Elastic Net} without the diagonal entry, using notation from \eqref{Block}:
\begin{equation*}
    -\textbf{w}_{12} +  \textbf{s}_{12}+ \lambda \alpha \Bo{\gamma}_{12} + \lambda (1-\alpha) \Bo{\theta}_{12}=\textbf{0} .
\end{equation*}
Use \eqref{WGlasso} for $\textbf{w}_{12}$ and define $\Bo{\beta} \coloneqq -\tfrac{\Bo{\theta}_{12}}{\theta_{22}}$:
\begin{equation}
    \begin{gathered}
     \Ma{W}_{11}\tfrac{\Bo{\theta_{12}}}{\theta_{22}}+ \textbf{s}_{12}+ \lambda \alpha \Bo{\gamma}_{12} + \lambda (1-\alpha) \Bo{\theta}_{12} =\textbf{0}, \\
       \Ma{W}_{11} \Bo{\beta} - \textbf{s}_{12} - \lambda \alpha \Bo{\gamma}_{12} + \lambda (1-\alpha) \theta_{22} \Bo{\beta} =\textbf{0} .
           \end{gathered} \label{Normal1}
\end{equation}
Note that $\Bo{\gamma}_{12} \in -\text{sign}(\Bo{\beta})$ since $\theta_{22} > 0$. Therefore \eqref{Normal1} corresponds to the normal equation of the following $L_1$ and $L_2$-regularized quadratic program:
    \begin{equation*}
            \underset{\Bo{\beta} \in \mathbb{R}^{p-1}}{\text{minimize }} \{ \tfrac{1}{2} \Bo{\beta}^T \textbf{W}_{11} \Bo{\beta} - \Bo{\beta}^T \textbf{s}_{12} + \lambda \alpha \norm{\Bo{\beta}}_1 + \lambda \tfrac{1-\alpha}{2} \theta_{22} \Bo{\beta}^T \Bo{\beta} \} ,
    \end{equation*}
or equivalently:
   \begin{equation}\label{QP2}
            \underset{\Bo{\beta} \in \mathbb{R}^{p-1}}{\text{minimize }} \{ \tfrac{1}{2} \norm{ \textbf{W}_{11}^{1/2} \Bo{\beta} - \textbf{W}_{11}^{-1/2}  \textbf{s}_{12} }_2^2 + \lambda \alpha \norm{\Bo{\beta}}_1 + \lambda \tfrac{1-\alpha}{2} \theta_{22} \norm{\Bo{\beta}}_2^2 \}  .
    \end{equation}
After solving the quadratic program (see Problem 2) below), with minimizer $\hat{\Bo{\beta}}$, update the entries as follows:
\begin{itemize}
    \item $\hat{\textbf{w}}_{12}= - \Ma{W}_{11} \frac{\Bo{\theta}_{12}}{\theta_{22}} =  \Ma{W}_{11} \hat{\Bo{\beta}}$ using the representation \eqref{WGlasso} and $\hat{\Bo{\beta}} =-\tfrac{\Bo{\theta}_{12}}{\theta_{22}}$
    \item $\hat{\theta}_{22} = \frac{1}{w_{22}-\hat{\Bo{\beta}} \hat{\Ma{w}}_{12} }$ by using \eqref{WGlasso}, $\hat{\Bo{\beta}} =-\tfrac{\Bo{\theta}_{12}}{\theta_{22}}$ and $\hat{\textbf{w}}_{12} =  \Ma{W}_{11} \hat{\Bo{\beta}}$
    \item $\hat{\Bo{\theta}}_{12}  =-\hat{\theta}_{22} \hat{\Bo{\beta}}$
     \item $\hat{w}_{22}$ with the normal equations \eqref{Norm: Elastic Net}
\end{itemize}

\paragraph{Solution to Problem 2)}  We show here how to apply componentwise updates for solving the $L_1$ and $L_2$-regularized quadratic program \eqref{QP2}. 
Using the notations $\Ma{Z}= \Ma{W}_{11}^{1/2}$, $\Ma{y}=\Ma{W}_{11}^{-1/2} \Ma{s}_{12}$, $\lambda_1=\lambda \alpha$ and $\lambda_2=\lambda (1-\alpha) \theta_{22}$ the quadratic program translates to a standard Elastic Net regression problem, i.e.,
   \begin{equation*}
      R(\Bo{\beta}) \coloneqq  \tfrac{1}{2} \norm{ \Ma{y} - \textbf{Z} \Bo{\beta} }_2^2 + \lambda_1 \norm{\Bo{\beta}}_1 + \tfrac{\lambda_2}{2} \norm{\Bo{\beta}}_2^2  
    \end{equation*}
needs to be minimized over $\Bo{\beta} \in \mathbb{R}^{p-1}$.     
Let $\Tilde{\beta}_k$ for $k \neq j$ be estimates and partially optimize $R(\Bo{\beta})$ with respect to $\beta_j$, by computing the gradient at $\beta_j=\Tilde{\beta_j}$. The gradient only exists if $\Tilde{\beta}_j \neq 0$. Without loss of generality assume that $\Tilde{\beta}_j > 0$. Then:
\begin{equation*}
    \frac{\partial R}{\partial \beta_j} \bigg|_{\Bo{\beta}=\Tilde{\Bo{\beta}}} =  - \sum_{i=1}^{p-1}  Z_{i,j}(y_i - Z_{i}^T \Tilde{\Bo{\beta}}) + \lambda_1 + \lambda_2 \Tilde{\beta}_j .
\end{equation*}
With the partial residuals given as $r_i^{(j)} = y_i - \sum_{k \neq j} Z_{ik} \Tilde{\beta}_k$ and setting the partial derivative to 0 yields:
\begin{equation*}
    \Tilde{\beta}_j = \frac{\sum_{i=1}^{p-1} Z_{i,j} r_i^{(j)}- \lambda_1 }{ \sum_{i=1}^{p-1} Z_{i,j}^2 + \lambda_2} .
\end{equation*}
A similar derivation can be done for $\Tilde{\beta}_j < 0 $. The case $\Tilde{\beta}_j=0$ is treated separately using standard sub-differential calculus. The overall solution for $\Tilde{\beta}_j$ satisfies:
\begin{equation}\label{QP: Updates}
    \Tilde{\beta}_j = \frac{ S_{\lambda_1} \left( \sum_{i=1}^{p-1} Z_{i,j} r_i^{(j)} \right) }{ \sum_{i=1}^{p-1} Z_{i,j}^2 + \lambda_2} ,
\end{equation}
where $S_{\lambda}(x) = \text{sign}(x)(\abs{x}- \lambda)_{+}$ is the soft-thresholding operator. Using the inner products $\Ma{Z}^T \Ma{Z} = \Ma{W}_{11}$ and $\Ma{Z}^T \Ma{y} = \Ma{s}_{12}$, equation \eqref{QP: Updates}
can be written as.
\begin{equation*}
    \Tilde{\beta}_j = \frac{  S_{\lambda_1} \left(  (\Ma{s}_{12})_j -\sum_{k \neq j} (\Ma{W}_{11})_{k,j} \Tilde{\beta}_k  \right) }{  (\Ma{W}_{11})_{j,j} + \lambda_2} .
\end{equation*}
Putting all these pieces from Problem 1) and Problem 2) together yields the \GN\ algorithm (Algorithm \ref{alg:algorithm3} below).

\begin{algorithm}[H]
    \caption{\GN\ algorithm}
  \begin{algorithmic}[1]
    \State Initialize $\Ma{\Theta}_{\text{init}}=\text{diag}(\textbf{S}+\lambda \alpha \Ma{I})^{-1}$ and $\textbf{W}_{\text{init}}= \textbf{S} + \lambda \alpha \Ma{I}  + \lambda (1- \alpha)\Ma{\Theta}_{\text{init}}$.
    \State Cycle around the columns repeatedly, performing the following steps till convergence:
\begin{algsubstates}
        \State Rearrange the rows/columns so that the currently updated column is last (implicitly).
        \State Solve the Elastic Net regression problem \eqref{QP2} with coordinate descent to get $\hat{\Bo{\beta}}$. As warm start for $\Bo{\beta}$ use the solution from the previous round for this row/column.
        \State Update $\hat{\Ma{w}}_{12} = \textbf{W}_{11} \hat{\Bo{\beta}}$, $\hat{\Ma{w}}_{21}=\hat{\Ma{w}}_{12}^T$.
        \State Update  $\hat{\theta}_{22} = \frac{1}{w_{22}-\hat{\Bo{\beta}} \hat{\Ma{w}}_{12} }$.
        \State Update $\hat{\Bo{\theta}}_{12}=- \hat{\theta}_{22} \hat{\Bo{\beta}}$, $\hat{\Bo{\theta}}_{21}=\hat{\Bo{\theta}}_{12}^T$.
        \State Update  $\hat{w}_{22}= \text{s}_{22}+\lambda \alpha+ \lambda (1-\alpha) \hat{\theta}_{22}$.
    \end{algsubstates}
  \end{algorithmic} 
  \label{alg:algorithm3}
\end{algorithm}

\subsubsection{\DPGN}
\paragraph{Solution to Problem 1)} 
Consider the $p$-th row of the normal equations in  \eqref{Norm: Elastic Net} without the diagonal entry using notation from \eqref{Block}:
\begin{equation*}
    -\textbf{w}_{12} +  \textbf{s}_{12}+ \lambda \alpha \Bo{\gamma}_{12} + \lambda (1-\alpha) \Bo{\theta}_{12}=\textbf{0} .
\end{equation*}
Use \eqref{WDPGlasso} for $\textbf{w}_{12}$ and $w_{22}$. Then multiply with $\Ma{\Theta}_{11}$ from the left to obtain
\begin{equation}
    \begin{gathered}
    \begin{aligned}
    \tfrac{\Ma{\Theta}^{-1}_{11}\Bo{\theta}_{12}}{\theta_{22}-\Bo{\theta}_{21}\Ma{\Theta}^{-1}_{11}\Bo{\theta}_{12}} +  \textbf{s}_{12}+ \lambda \alpha \Bo{\gamma}_{12} + \lambda (1-\alpha) \Bo{\theta}_{12} =\textbf{0}, \\
    \Ma{\Theta}^{-1}_{11} w_{22} \Bo{\theta}_{12} +  \textbf{s}_{12}+ \lambda \alpha \Bo{\gamma}_{12} + \lambda (1-\alpha) \Bo{\theta}_{12} =\textbf{0}, \\
     w_{22} \Bo{\theta}_{12} +  \Ma{\Theta}_{11}(\textbf{s}_{12}+ \lambda \alpha \Bo{\gamma}_{12} + \lambda (1-\alpha) \Bo{\theta}_{12}) =\textbf{0}.
     \end{aligned}
    \end{gathered}\label{Normal2}
\end{equation}
Consider the case with $\alpha \in (0,1]$. Define $\Tilde{\Bo{\gamma}} \coloneqq  \lambda \alpha \Bo{\gamma}_{12}$ as well as $\Tilde{\Ma{q}}_{12} \coloneqq \text{abs}(\Bo{\theta}_{12})$ and get the following Karush-Kuhn-Tucker (KKT) conditions:
\begin{equation}
    \begin{gathered}
    \tfrac{w_{22}}{\lambda \alpha} \Tilde{\Bo{q}}_{12} * \Tilde{\Bo{\gamma}} + \Ma{\Theta}_{11}(\textbf{s}_{12}+(1+\tfrac{1-\alpha}{\alpha}\Tilde{\Bo{q}}_{12}) * \Tilde{\Bo{\gamma}}) = \textbf{0},  \\ 
   \Tilde{\Bo{q}}_{12} * (\text{abs}(\Tilde{\Bo{\gamma}}) - \lambda \alpha 1_{p-1})= \textbf{0}, \\
   \norm{\Tilde{\Bo{\gamma}}}_{\infty} \leq \lambda \alpha, 
    \end{gathered} \label{Normal Row2} 
\end{equation}
where $*$ denotes elementwise multiplication. These are equivalent to the following box-constrained quadratic program for $\Bo{\gamma} \in \mathbb{R}^{p-1}$ :
\begin{equation} \label{QP3}
    \underset{\norm{\Bo{\gamma}}_{\infty} \leq \lambda \alpha }{\text{minimize }} \{ \tfrac{1}{2}(\textbf{s}_{12}+(1+\tfrac{1-\alpha}{\alpha}\Tilde{\Bo{q}}_{12}) * \Bo{\gamma})^T \Ma{\Theta}_{11} (\textbf{s}_{12}+(1+\tfrac{1-\alpha}{\alpha}\Tilde{\Bo{q}}_{12}) * \Bo{\gamma}) \} .\\
\end{equation}
In the special case where $\alpha=0$, the normal equations \eqref{Normal2} simplify to:
\begin{equation*}
    w_{22} \Bo{\theta}_{12} +  \Ma{\Theta}_{11}(\textbf{s}_{12} + \lambda \Bo{\theta}_{12})=\textbf{0} .
\end{equation*}
Define $\Tilde{\Bo{q}}_{12}=\text{abs}(\Bo{\theta}_{12})$ to get the following Karush-Kuhn-Tucker (KKT) conditions:
\begin{gather*}
    w_{22} \Tilde{\Bo{q}}_{12} * \Bo{\gamma} + \Ma{\Theta}_{11}(\textbf{s}_{12}+\lambda \Tilde{\Bo{q}}_{12} * \Bo{\gamma}) = \textbf{0}, \\
   \Tilde{\Bo{q}}_{12} * (\text{abs}(\Bo{\gamma}) - 1_{p-1})= \textbf{0}, \\
   \norm{\Bo{\gamma}}_{\infty} \leq 1.
\end{gather*}
These are equivalent to the following box-constrained quadratic program for $\Bo{\gamma} \in \mathbb{R}^{p-1}$:
\begin{equation*} \label{QP4}
   \underset{ \norm{\Bo{\gamma}}_{\infty} \leq 1 }{\text{minimize }}  \tfrac{1}{2}(\textbf{s}_{12}+\lambda \Tilde{\Bo{q}}_{12}* \Bo{\gamma})^T \Ma{\Theta}_{11} (\textbf{s}_{12}+\lambda \Tilde{\Bo{q}}_{12}* \Bo{\gamma}) .
\end{equation*}
After solving the quadratic program (see Problem 2) below) with solution $\Bo{\gamma}^*$, update for $\alpha \in (0,1]$  in the following way (and with very similar updates for the case $\alpha=0$): 
\begin{itemize}
    \item $\hat{\Bo{\theta}}_{12} = \tfrac{- \Ma{\Theta}_{11} (\textbf{s}_{12}+(1+\tfrac{1-\alpha}{\alpha}\Tilde{\Bo{q}}_{12}) * \Bo{\gamma}^*)}{w_{22}}$ using \eqref{Normal Row2} and that $\hat{\Bo{\theta}}_{12}= \Tilde{\Bo{q}}_{12} * \Tilde{\Bo{\gamma}} $
    \item $\hat{\theta}_{22} = w_{22} + \hat{\Bo{\theta}}_{12}^T  \Bo{\Theta}^{-1} \hat{\Bo{\theta}}_{12}$ by \eqref{WDPGlasso} and then use $\hat{\Bo{\theta}}_{12}$ to get $\hat{\theta}_{22} = \tfrac{1- (\textbf{s}_{12}+(1+\tfrac{1-\alpha}{\alpha}\Tilde{\Bo{q}}_{12}) * \Bo{\gamma}^*)^T\hat{\Bo{\theta}}_{12}}{w_{22}}$ 
    \item $\hat{\textbf{w}}_{12}$ and $\hat{w}_{22}$ with the normal equations \eqref{Norm: Elastic Net}
\end{itemize}
\paragraph{Solution to Problem 2)}  We show here how to apply componentwise updates for solving the box constrained quadratic program  \eqref{QP3}. Taking the derivative of the $j$-th coordinate with respect to $\textbf{c}$ for the quadratic program of the form $\tfrac{1}{2}(\textbf{a} + \textbf{b} * \textbf{c})^T \textbf{D}(\textbf{a} + \textbf{b} * \textbf{c})$ leads to:
\begin{equation*}
    \begin{split}
    [\tfrac{1}{2}(\textbf{a} + \textbf{b} * \textbf{c})^T \textbf{D}(\textbf{a} + \textbf{b} * \textbf{c})]^{(j)} &= 
    \tfrac{1}{2}[\textbf{a}^T \textbf{D}\textbf{a}]^{(j)} +
     [(\textbf{b} * \textbf{c})^T\textbf{D}\textbf{a}]^{(j)} +
   \tfrac{1}{2}[(\textbf{b} * \textbf{c})^T \textbf{D}(\textbf{b} * \textbf{c})]^{(j)} \\
    &=
    0 +(\textbf{b} * (\textbf{D}\textbf{a}))_j + (\textbf{b} * (\textbf{D}(\textbf{b} * \textbf{c})))_j \\
    &= (\textbf{b} * (\textbf{D}(\textbf{a}+ \textbf{b} * \textbf{c})))_j \\
    &= b_j \sum_{k=1}^{p} d_{j,k} (a_k + b_k c_k) . 
    \end{split}
\end{equation*}
Setting $b_j \sum_{k=1}^{p}  d_{j,k} (a_k + b_k c_k) =0$ one obtains 
\begin{equation*}
\begin{split}
 d_{j,j} (a_j + b_j c_j) &= \sum_{k=1, k \neq j }^{p} d_{j,k} (a_k + b_k c_k), \\
 c_j &= \frac{-\sum_{k=1}^{p} d_{j,k} (a_k + b_k c_k) + d_{j,j} b_j c_j}{d_{j,j} b_j }.
\end{split}
\end{equation*}
This leads to the following coordinate update of $\gamma_j$
\begin{equation*}
    \gamma_{j,new} = \frac{-\sum_{k=1}^{p} (\Ma{\Theta}_{11})_{j,k} (\textbf{s}_{12}+(1+\tfrac{1-\alpha}{\alpha}\Tilde{\Bo{q}}_{12})_k * \Ma{\gamma}_{k, old}) + (\Ma{\Theta}_{11})_{j,j}(1+\tfrac{1-\alpha}{\alpha}\Tilde{\Bo{q}}_{12})_j * \Ma{\gamma}_{j, old}}{ (\Ma{\Theta}_{11})_{j,j}(1+\tfrac{1-\alpha}{\alpha}\Tilde{\Bo{q}}_{12})_j } .
\end{equation*}
In the special case where $\alpha=0$,
\begin{equation*}
    \gamma_{j,new} = \frac{-\sum_{k=1}^{p} (\Ma{\Theta}_{11})_{j,k} (\textbf{s}_{12}+\lambda (\Tilde{\Bo{q}}_{12})_k * \Ma{\gamma}_{k, old}) + (\Ma{\Theta}_{11})_{j,j}\lambda (\Tilde{\Bo{q}}_{12})_j * \Ma{\gamma}_{j, old}}{ (\Ma{\Theta}_{11})_{j,j}\lambda(\Tilde{\Bo{q}}_{12})_j } .
\end{equation*}
Putting all these pieces from Problems 1) and 2) together yields the \DPGN\ algorithm. Algorithm \ref{alg:algorithm4} below shows the procedure for $\alpha \in (0,1]$ (which changes only slightly for the case $\alpha= 0$).

\begin{algorithm}[H]
  \caption{\DPGN\ algorithm}
  \begin{algorithmic}[1]
    \State Initialize $\Ma{\Theta}_{\text{init}}=\text{diag}(\lambda \alpha \textbf{I} + \textbf{S})^{-1}$ and $\textbf{W}_{\text{init}}= \textbf{S} + \lambda \alpha \textbf{I} + \lambda (1- \alpha)\Ma{\Theta}_{\text{init}}$.
    \State Cycle around the columns repeatedly, performing the following steps till convergence:
    \begin{algsubstates}
        \State Rearrange the rows/columns so that the currently updated column is last (implicitly).
        \State Solve \eqref{QP3} and denote the solution as $\Bo{\gamma}^*$.
        \State Update $\hat{\Bo{\theta}}_{12} = \tfrac{- \Ma{\Theta}_{11} (\textbf{s}_{12}+(1+\tfrac{1-\alpha}{\alpha}\Tilde{\Bo{q}}_{12}) * \Bo{\gamma}^*)}{w_{22}}$.
        \State Update $\hat{\theta}_{22} = \tfrac{1- (\textbf{s}_{12}+(1+\tfrac{1-\alpha}{\alpha}\Tilde{\Bo{q}}_{12}) * \Bo{\gamma}^*)^T\hat{\Bo{\theta}}_{12}}{w_{22}}$.
        \State Update $\hat{\textbf{w}}_{12}  = \textbf{s}_{12}+ \Bo{\gamma}^* + \lambda (1-\alpha) \hat{\Bo{\theta}}_{12}$.
        \State Update $\hat{w}_{22}  = s_{22}+ \lambda \alpha + \lambda (1-\alpha) \hat{\theta}_{22}$.
    \end{algsubstates}
  \end{algorithmic} 
  \label{alg:algorithm4}
\end{algorithm}

\begin{lemma}
\label{Lemma:DPGN}
Suppose $\Ma{\Theta} \succ \Ma{0}$ is used as warm-start for the \DPGN\ algorithm. Then every row/column update of \DPGN\ maintains the positive definiteness of the working precision matrix $\Ma{\Theta}$. Note that a corresponding lemma for the \DPGL\ is proven by \cite{Mazumder2012_dpglasso} and hence, we just provide a simple extension here.
\end{lemma}
\begin{proof}
Let $\Ma{A}= \begin{pmatrix}
    \textbf{A}_{11} & \textbf{a}_{12} \\
    \textbf{a}_{21} & a_{22} 
  \end{pmatrix}$. The condition $\Ma{A} \succ \Ma{0}$ is equivalent to: \\
  \begin{equation}\label{rem1}
       \textbf{A}_{11}  \succ \Ma{0} \text{ and } (a_{22} - \textbf{a}_{21} (\textbf{A}_{11})^{-1} \textbf{a}_{12} ) >0 .
  \end{equation}
Consider updating the $p$-th row/column of the precision matrix. Since the block $\Ma{\Theta}_{11}$ remains fixed we only need to show the second condition from \eqref{rem1}. Using the updates of the \DPGN\ (Algorithm \ref{alg:algorithm4}):
\begin{gather*}
\hat{\theta}_{22} - \hat{\Bo{\theta}}_{12}^T (\Ma{\Theta}_{11})^{-1} \hat{\Bo{\theta}}_{12}  \\
=\frac{1- (\textbf{s}_{12}+(1+\tfrac{1-\alpha}{\alpha}\Tilde{\Bo{q}}_{12}) * \Bo{\gamma}^*)^T\hat{\Bo{\theta}}_{12}}{w_{22}}  
- \frac{- (\textbf{s}_{12}+(1+\tfrac{1-\alpha}{\alpha}\Tilde{\Bo{q}}_{12}) * \Bo{\gamma}^*)^T (\Ma{\Theta}_{11})^{-1} \Ma{\Theta}_{11} \hat{\Bo{\theta}}_{12} }{w_{22}} \\
= \frac{1}{w_{22}} = \frac{1}{s_{22} + \lambda \alpha + \lambda (1-\alpha) \theta_{22}} > 0 .  
\end{gather*}
\end{proof}

\subsubsection{Connected components} \label{subsec:connectedcomp}
The problem of solving $p$-dimensional \GL\ (or \DPGL) problems can be reduced in some cases to solving several lower-dimensional problems, enabling massive speed ups for the computations (see \citealp{Mazumder2012_connected_components, Witten2011_connected_components}). A similar result also holds for \GN\ (or \DPGN). Following the notation of \cite{Mazumder2012_connected_components}, define the nodes $\mathcal{V}=\{1,...,p \}$ and the matrix~$\mathcal{E}$: 
\begin{equation*}
    \mathcal{E}_{ij}= \begin{cases} 1 \text{ if } \Ma{\hat\Theta}_{ij}(\lambda, \alpha) \neq 0, i \neq j, \\
    0 \text{ otherwise.}
    \end{cases}
\end{equation*}
Here, $\Ma{\hat\Theta}(\lambda, \alpha)$ is the estimated precision matrix for the input covariance matrix $\textbf{S}$ and the two tuning parameters $\lambda$ and $\alpha$, see equation \eqref{problem_setup}. This defines a symmetric graph $\mathcal{G}=(\mathcal{V}, \mathcal{E})$. Decompose this graph into its connected components, i.e., $\mathcal{G}= \bigcup_{l=1}^{\tilde{L}} \mathcal{G}_{l}$ where $\tilde{L}$ is the number of connected components and $\mathcal{G}_{l}=({\mathcal{V}}_l, \mathcal{E}_l)$ the $l$-th sub-graph. Furthermore, define $E$ as
\begin{equation*}
    E_{ij}= \begin{cases} 1 \text{ if } \abs{\textbf{S}_{ij}} > \lambda\alpha, i \neq j \\
    0 \text{ otherwise.}
    \end{cases}
\end{equation*}
This defines another symmetric graph $G=(\mathcal{V},E)$. Decompose this graph into its connected components as well, i.e., $G= \bigcup_{l=1}^{L} G_{l}$ where $L$ is the number of connected components and $G_{l}=(\mathcal{V}_l, E_l)$.

\begin{theorem}[taken from \citealp{Atchade_Mazumder_Chen}]
\label{thm:connected_components}
Let $\mathcal{G}_{l}=({\mathcal{V}}_l, \mathcal{E}_l), l=1,...,\tilde{L}$ and  $G_{l}=(\mathcal{V}_l, E_l), l=1,...,L$ denote the connected components, as defined above.
Then $\tilde{L}=L$ and there exists a permutation $\Pi$ on $\{1,...,L\}$ such that ${\mathcal{V}}_{\Pi(l)}=\mathcal{V}_l$ and ${\mathcal{E}}_{\Pi(l)}=E_l$ for all $l=1,...,L$.
\end{theorem}
Determining the connected components based on the thresholded covariances (as defined in $E$) is computationally cheap. Note that parts of $\Ma{\hat\Theta}$ corresponding to the different connected component can be then solved independently, i.e. a suitable permutation of $\Ma{\Theta}$ and $\Ma{W}$ leads to block-diagonal form. The $L$ blocks are exactly of the sizes of the connected components. This result is especially attractive if the maximum size of the connected components is small compared to $p$, since the additional effort to compute the connected components is negligible compared to the gains of reducing the problem into smaller sized problems. For fixed $\alpha \neq 0$ and $\Ma{S}$ the number of connected components is increasing in $\lambda$. If $\lambda \alpha \geq S_{ij}$ for all $i,j \in \mathcal{V}$ with $i \neq j$  then $\Ma{\Theta}$ and $\Ma{W}$ are diagonal matrices. We incorporated the check of such connected components prior to starting actual calculations in our implementation in the \oursoftware\ \Rsoftware-package.

\subsection{Target matrices}
\label{chp:target}
We now turn towards the inclusion of a positive semi-definite (diagonal) target matrix $\textbf{T}$ into the Graphical Elastic Net problem. Some motivation for this is provided by \cite{Wieringen_Peeters} as well as \cite{Kuismin_ROPE} in the case of Ridge type penalization. The target can be interpreted as prior knowledge or an educated guess. We aim to solve problem \eqref{problem_setup} which we recall here:
\begin{equation*}
       \hat{\mathbf{\Theta}}(\lambda, \alpha, \textbf{T}) = 
    \underset{\mathbf{\Theta} \succ 0 }{\text{argmin}}   \{ -\text{log det} \Ma{\Theta} + \text{tr} (\textbf{S} \Ma{\Theta}) + \lambda (\alpha \norm{\Ma{\Theta}- \textbf{T}}_1+\tfrac{1-\alpha}{2} \norm{\Ma{\Theta}-\textbf{T}}_2^2) \} .
\end{equation*}
This optimization problem is difficult to solve efficiently in general. In the $L_1$-penalty case,  \cite{Wieringen_iterative} proposed to iteratively apply his generalized Ridge estimator (with elementwise differing $\lambda$) to approximate the loss function for the $L_1$ case. However, this approach is computationally not attractive, even for moderately sized problems (see section~\ref{chp:comptime} for computational times). We propose to simplify the problem by considering only positive semi-definite \textbf{diagonal} target matrices. While diagonal matrices are less flexible than arbitrary positive semi-definite target matrices, we note that in practice many fall into this category (e.g. all target matrices that were used by \cite{Wieringen_Peeters} and by \cite{Kuismin_ROPE}).

When considering diagonal target matrices with non-negative entries, the normal equations for the diagonal entries are different, while for the other entries they remain the same. For each diagonal entry three cases can occur, leading to the following normal equations:\\
Case 1: $\theta_{22} > t_{22}$,  then
\begin{equation*}
 w_{22} = s_{22} + \lambda \alpha + \lambda (1-\alpha) (\theta_{22}-t_{22})  .
\end{equation*}
Case 2: $\theta_{22} = t_{22}$,  then
\begin{equation*}
 w_{22} = s_{22} + \lambda \alpha u, \text{ where } u \in [-1,1]. 
\end{equation*}
Case 3: $\theta_{22} < t_{22}$,  then
\begin{equation*}
 w_{22} = s_{22} - \lambda \alpha + \lambda (1-\alpha) (\theta_{22}-t_{22}) .
\end{equation*} 
Note that in the traditional setting when zero is the target matrix, then always case~1 occurs. In the general case, however, one cannot automatically update diagonal elements based on case~1. We derive the technical modifications in Appendix~\ref{app:target} on how to modify the \GN\ algorithm. Overall, in contrast to the algorithm without target, the initialization and updating for the diagonal entries require care. We note that our chosen updates work for realistic targets, but occasionally may not converge for targets with very large diagonal entries. However, such targets are not desirable from a statistical perspective.

\subsection{Discussion of the different algorithms}

Our proposed algorithms \GN\ and \DPGN\ are generalizations of \GL\ and \DPGL. For the Lasso case $\alpha = 1$, \GN\ performs the same updates as \GL\ and \DPGN\ performs the same updates as \DPGL. For elastic net penalties, i.e.~$\alpha \in (0,1)$, updates of \GN\ and \DPGN\ are slightly modified to incorporate the additional quadratic penalty term. Lemma~\ref{Lemma:DPGN} ensures that \DPGN\ maintains positive definite precision matrices throughout the algorithm for arbitrary positive definite precision matrices as warm starts. This feature is appealing when a fit is already available, for example from a slightly different $\lambda$ value in cross validation or in change point detection problems from a neighbouring split point (see \citealp{SeedBS} and \citealp{OBS}). \GN\ lacks this feature, and for certain warm starts it may not converge. However, as an advantage for \GN, we have updates that allow to incorporate diagonal target matrices, and hence, is more flexible than \DPGN in its currently presented form. As long as single fits are required, (i.e.~without warm starts), we recommend \GN. If repeated fits relying on warms starts are necessary in some application, one can still try to use \GN\ with warm starts and in case this faces convergence issues, resort to \DPGN.

\section{Simulation results} \label{chp:simulation}

In this chapter the statistical performance of the Graphical Elastic Net estimator is compared to its special cases \GL\ \citep{glasso} and \RP\ \citep{Kuismin_ROPE}. Note that \cite{Wieringen_Peeters} also proposed Ridge-type penalization and they called it the Alternative Ridge Precision estimator (with Type I for the case of zero as target matrix and Type II for nonzero positive definite target matrix). For simplicity, we will use the name \RP\ in the following for such Ridge-type penalization approaches.

\paragraph{Simulation models.}
The simulation setup is similar to that of \cite{Kuismin_ROPE}. We draw $n$ independent realizations from multivariate Gaussian distributions $\mathcal{N}(\Ma{0},\Ma{\Sigma})$, where $\Ma{\Sigma}= [\sigma_{i,j}] \in \mathbb{R}^{p \times p}$ and $\Ma{\Theta}= [\theta_{i,j}]= \Ma{\Sigma}^{-1}$ are positive definite matrices coming from 6 different models.
\begin{itemize}
    \item \textit{Model 1} \textbf{Compound symmetry model}: $\sigma_{i,i}=1$ and $\sigma_{i,j}=0.6^2$ for $i\neq j$.
 \end{itemize}
 Model 2-4 are taken from \cite{Liu_Wang_TIGER}. In these models an adjacency matrix $\mathbf{A}$ is generated from a graph, where each nonzero off-diagonal element is set to 0.3 and the diagonal elements to 0.  Then the smallest eigenvalue $\Lambda_{\text{min}}(\mathbf{A})$ is calculated and the corresponding precision matrix is constructed by
\begin{equation*}
    \mathbf{\Theta} = \mathbf{D} ( \mathbf{A} + ( \abs{ \Lambda_{\text{min}}( \mathbf{A})}+0.2) \cdot \mathbf{I}) \mathbf{D},
\end{equation*}
where $\mathbf{D} \in \mathbb{R}^{p \times p}$ is a diagonal matrix with $D_{i,i}=1$ for $i=1,\dots,\tfrac{p}{2}$ and $D_{i,i}=3$ for $i=\tfrac{p}{2}+1,\dots,p$.
\begin{itemize}
    \item \textit{Model 2} \textbf{Scale-free graph model}: The graph begins with an initial small chain graph of 2 nodes. New nodes are added to the graph one at a time. Each new node is connected to one existing node with a probability proportional to the number of degrees that the existing node already has.
    \item \textit{Model 3} \textbf{Hub graph model}: The $p$ nodes are evenly partitioned into $\tfrac{p}{10}$ disjoint groups with each group containing 10 nodes. Within each group, one node is selected as the hub and edges between the hub and the other 9 nodes are added.
    \item \textit{Model 4} \textbf{Block graph model}: Here $\tilde{\mathbf{\Theta}}$ is directly produced by making a block diagonal matrix with block size $\tfrac{p}{10}$, where the off-diagonal entries are set to 0.5 and diagonal entries to 1. The matrix is then randomly permuted by rows/columns and the resulting covariance matrix is taken as $\mathbf{\Sigma} = \mathbf{D}^{-1} \tilde{\mathbf{\Theta}}^{-1} \mathbf{D}^{-1}$, where this time $D_{i,i}=1$ for $i=1,\dots,\tfrac{p}{2}$ and $D_{i,i}=1.5$ for $i=\tfrac{p}{2}+1,\dots, p$.
\end{itemize}
 Models 5 \& 6 are graphical models coming from \cite{Cai2}. First generate the matrix $\tilde{\mathbf{\Theta}}=[\tilde{\theta}_{i,j}]$ and then multiply with the inverse of a diagonal matrix $\mathbf{D}$ from both sides to get $\mathbf{\Sigma}$. Each diagonal entry of $\mathbf{D}$ is independently generated from a uniform distribution on the interval 1 to 5.
 \begin{itemize}
    \item \textit{Model 5} \textbf{Band graph model}: $\tilde{\theta}_{i,i}=1$, $\tilde{\theta}_{i,i+1}=\tilde{\theta}_{i+1,i}=0.6$, $\tilde{\theta}_{i,i+2}=\tilde{\theta}_{i+2,i}=0.3$, $\tilde{\theta}_{i,j}=0$ for $\abs{i-j}\geq 3$.
    \item \textit{Model 6} \textbf{Erd\H{o}s-R\'enyi random graph model}: Take $\tilde{\tilde{\mathbf{\Theta}}}=[\tilde{\tilde{\theta}}_{i,j}]$, where $\tilde{\tilde{\theta}}_{i,j}=u_{i,j} \cdot \delta_{i,j}$, such that $\delta_{i,j}$ is a Bernoulli random variable with success probability 0.05 and $u_{i,j}$ is a uniform random variable on the interval 0.4 to 0.8. Then take $\tilde{\mathbf{\Theta}}=\tilde{\tilde{\mathbf{\Theta}}} + ( \lvert  \Lambda(\tilde{\tilde{\mathbf{\Theta}}})_{\text{min}} \rvert + 0.05) \cdot \mathbf{I} $.
 \end{itemize}

For all models we transform the covariance matrix to be a correlation matrix before simulating the data. With the exception of Model 1, all models have a sparse structure in $\mathbf{\Theta}$. 
In order to get performance measures for the different methods, 100 independent simulations for each model are performed and the averages for several loss functions are calculated.
Using five-fold cross-validation the optimal tuning parameter $\lambda$ for each method is determined for each simulation run separately.

\paragraph{Performance measures.}
The different measures used in the simulations can be split into two groups. \\
1) Loss functions used by \cite{Kuismin_ROPE}:
\begin{itemize}
    \item \textit{Kullback-Leibler loss:} KL = tr($\Ma{\Sigma} \hat{\Ma{\Theta}})$ - log(det($\Ma{\Sigma} \hat{\Ma{\Theta}})) - p$
    \item \textit{L2 loss:} L2 $= \lVert \Ma{\Theta} - \hat{\Ma{\Theta}} \rVert_F$
    \item \textit{Spectral norm loss:} SP = $d_1$, where $d_1^2$ is the largest eigenvalue of the matrix $(\Ma{\Theta}-\hat{\Ma{\Theta}})^T(\Ma{\Theta}-\hat{\Ma{\Theta}})$
\end{itemize}
2) Graph recovery measures: 
Let $\hat{\Ma{\Theta}}$ be the estimated solution and $\Ma{\Theta}$ the underlying truth. For a given threshold $\epsilon$ define the adjacency matrices $\mathcal{\hat A}$ and $\Ma{A}$:
\begin{equation*}
    \mathcal{\hat A}_{i,j}= \begin{cases} 1 \text{ if } \hat{\Ma{\Theta}}_{i,j} \geq \epsilon \\
    0 \text{ otherwise.}
    \end{cases}
    \Ma{A}_{i,j}= \begin{cases} 1 \text{ if } \Ma{\Theta}_{i,j} \geq \epsilon \\
    0 \text{ otherwise.}
    \end{cases}
\end{equation*}
For each $i<j$ the edge $ij$ in $\mathcal{\hat A}$ is present if $\mathcal{\hat A}_{i,j}=1$, and similar for $\Ma{A}$. Now define the following quantities.
\begin{itemize}
    \item \textit{True Positives:} TP = number of edges $ij$, which are both in $\mathcal{\hat A}$ and in $\Ma{A}$
    \item \textit{True Negatives:} TN = number of edges $ij$, which are both not in $\mathcal{\hat A}$ and not in $\Ma{A}$
    \item \textit{False Positives:} FP = number of edges $ij$ in  $\mathcal{\hat A}$ but not in $\Ma{A}$
    \item \textit{False Negatives:} FN = number of edges $ij$ in $\Ma{A}$ but not in $\mathcal{\hat A}$
\end{itemize}
Then define the graph recovery measures.
\begin{itemize}
    \item \textit{F1 score:} F1 $= \tfrac{2 \text{TP}}{ 2 \text{TP+FN+FP}}$
    \item \textit{Matthews correlation coefficient:} MCC $= \tfrac{\text{TP} \times \text{TN - FP} \times \text{FN}}{\sqrt{ \text{(TP+FP)(TP+FN)(TN+FP)(TN+FN) }}}$ 
\end{itemize}
Note that both of the latter measures take values in $[0,1]$ with 0 being the worst and 1 being the best value. The threshold $\epsilon$ in the simulations is fixed as $\epsilon=10^{-5}$. When interpreting later on the results shown in Figure~\ref{fig:F1} and Figure~\ref{fig:MCC}, one should keep in mind that the \RP\ algorithm produces non-sparse solutions and therefore almost no true negatives and false negatives are produced for $\epsilon=10^{-5}$. Moreover, Model~1 is not sparse and thus analyzing the graph recovery measures in this model is not meaningful.

\begin{figure}[H]
	\centering
  \includegraphics[width=0.9\textwidth]{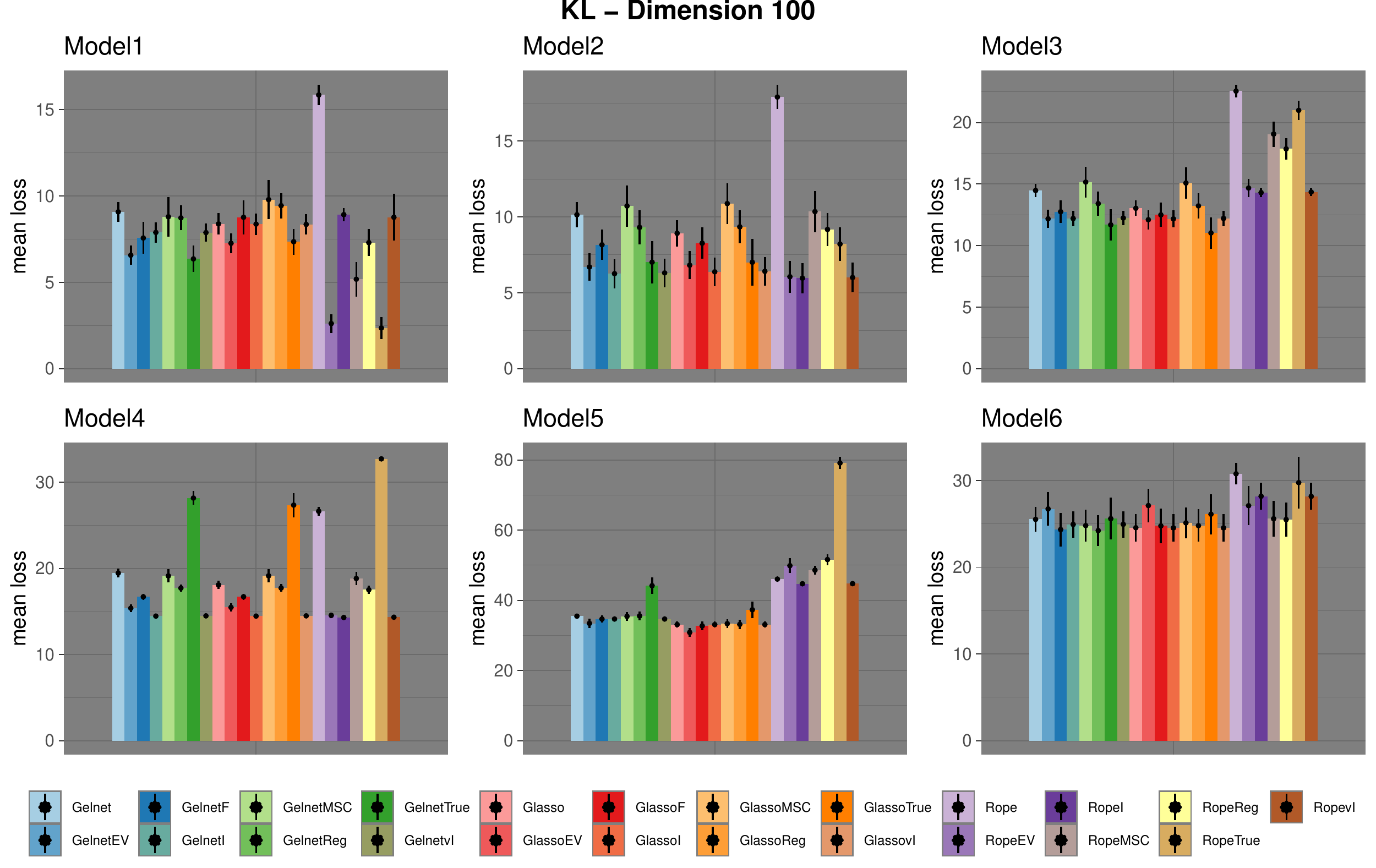}
  \caption{Summary of KL loss for the different models and different methods based on 100 replications. The columns along the small black dots indicate mean losses. The bars on the top of each column show the standard deviations (mean $\pm$ SD). Plot layout is taken from \cite{Kuismin_ROPE}.}
  \label{fig:KL}
\end{figure}

\begin{figure}[H]
	\centering
  \includegraphics[width=0.9\textwidth]{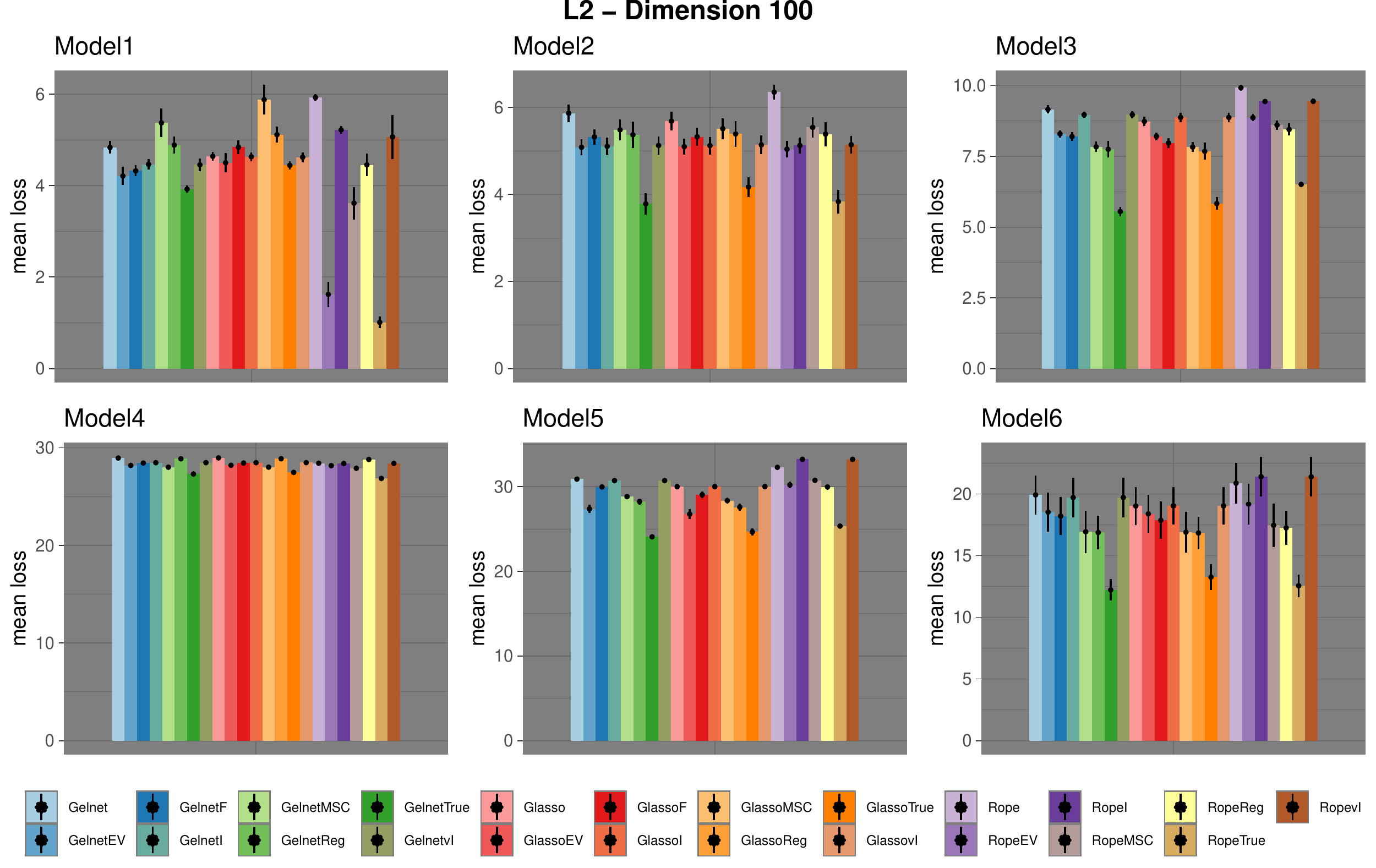}
  \caption{Summary of L2 loss for the different models and different methods based on 100 replications. The columns along the small black dots indicate mean losses. The bars on the top of each column show the standard deviations (mean $\pm$ SD). Plot layout is taken from \cite{Kuismin_ROPE}.}
    \label{fig:L2}
\end{figure}

\begin{figure}[H]
	\centering
  \includegraphics[width=0.9\textwidth]{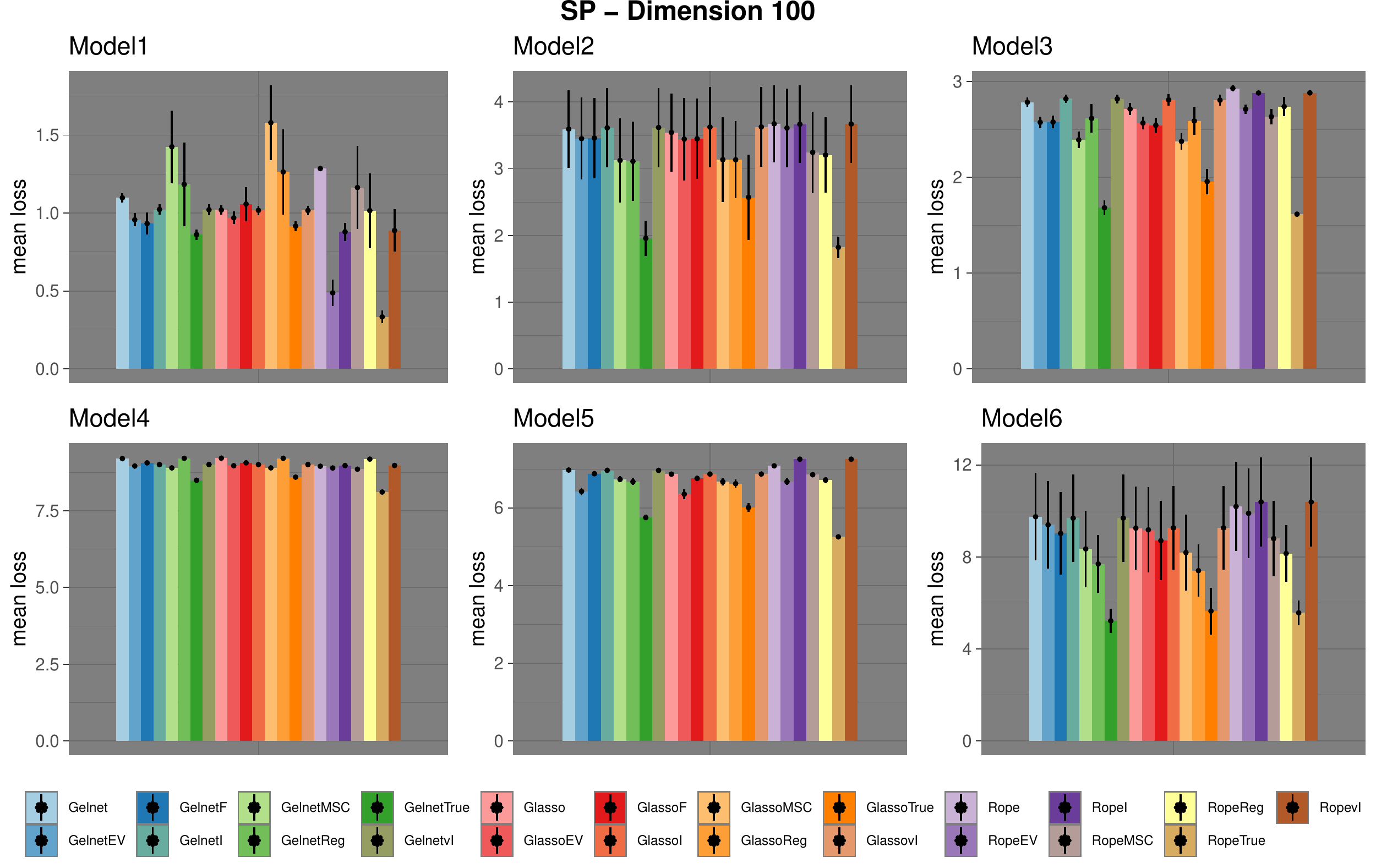}
  \caption{Summary of SP loss for the different models and different methods based on 100 replications. The columns along the small black dots indicate mean losses. The bars on the top of each column show the standard deviations (mean $\pm$ SD). Plot layout is taken from \cite{Kuismin_ROPE}.}
    \label{fig:SP}
\end{figure}

\begin{figure}[H]
	\centering
  \includegraphics[width=0.9\textwidth]{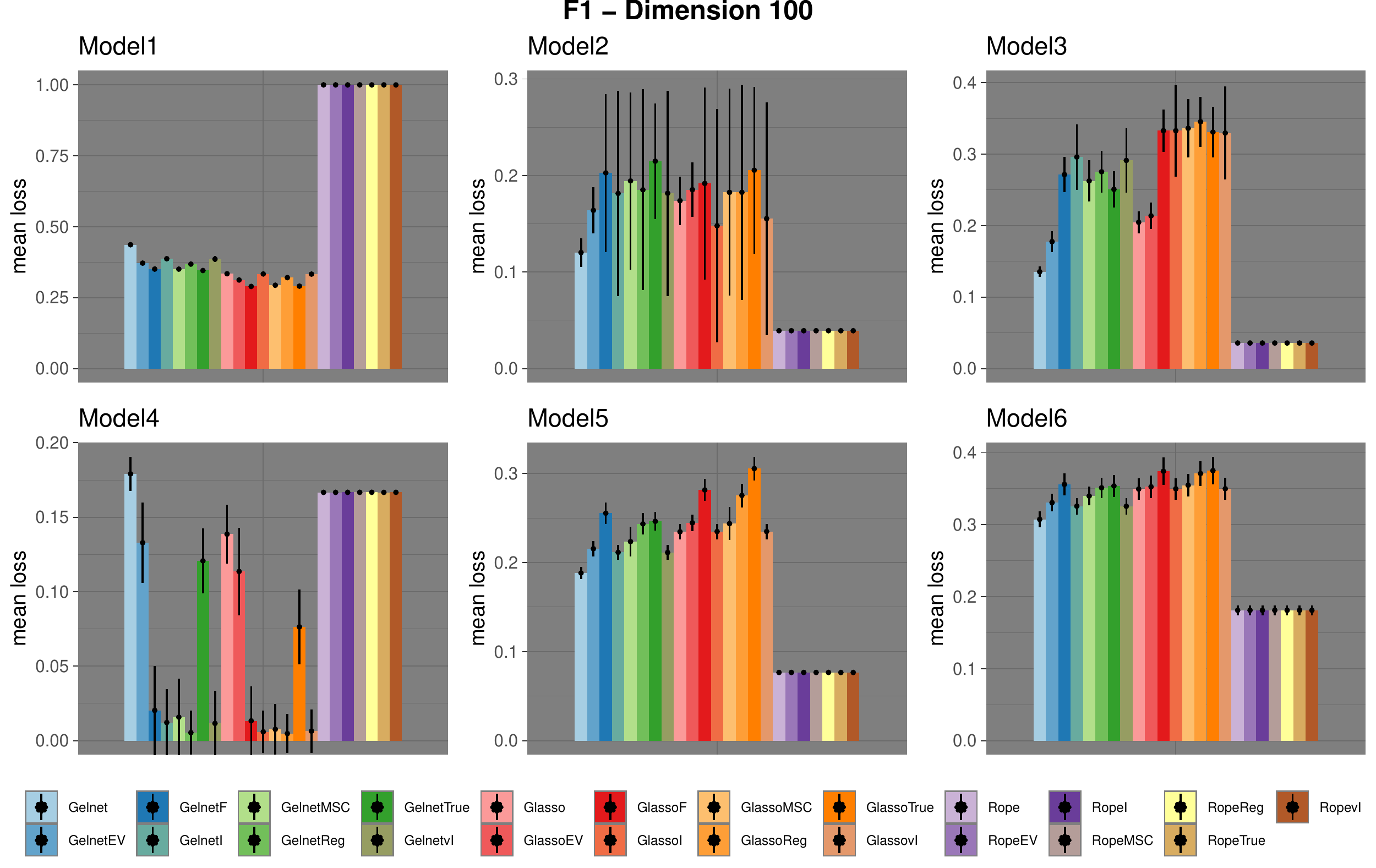}
  \caption{Summary of F1 score for the different models and different methods based on 100 replications. The columns along the small black dots indicate mean scores. The bars on the top of each column show the standard deviations (mean $\pm$ SD). Plot layout is taken from \cite{Kuismin_ROPE}.}
    \label{fig:F1}
\end{figure}
\begin{figure}[H]
	\centering
  \includegraphics[width=0.9\textwidth]{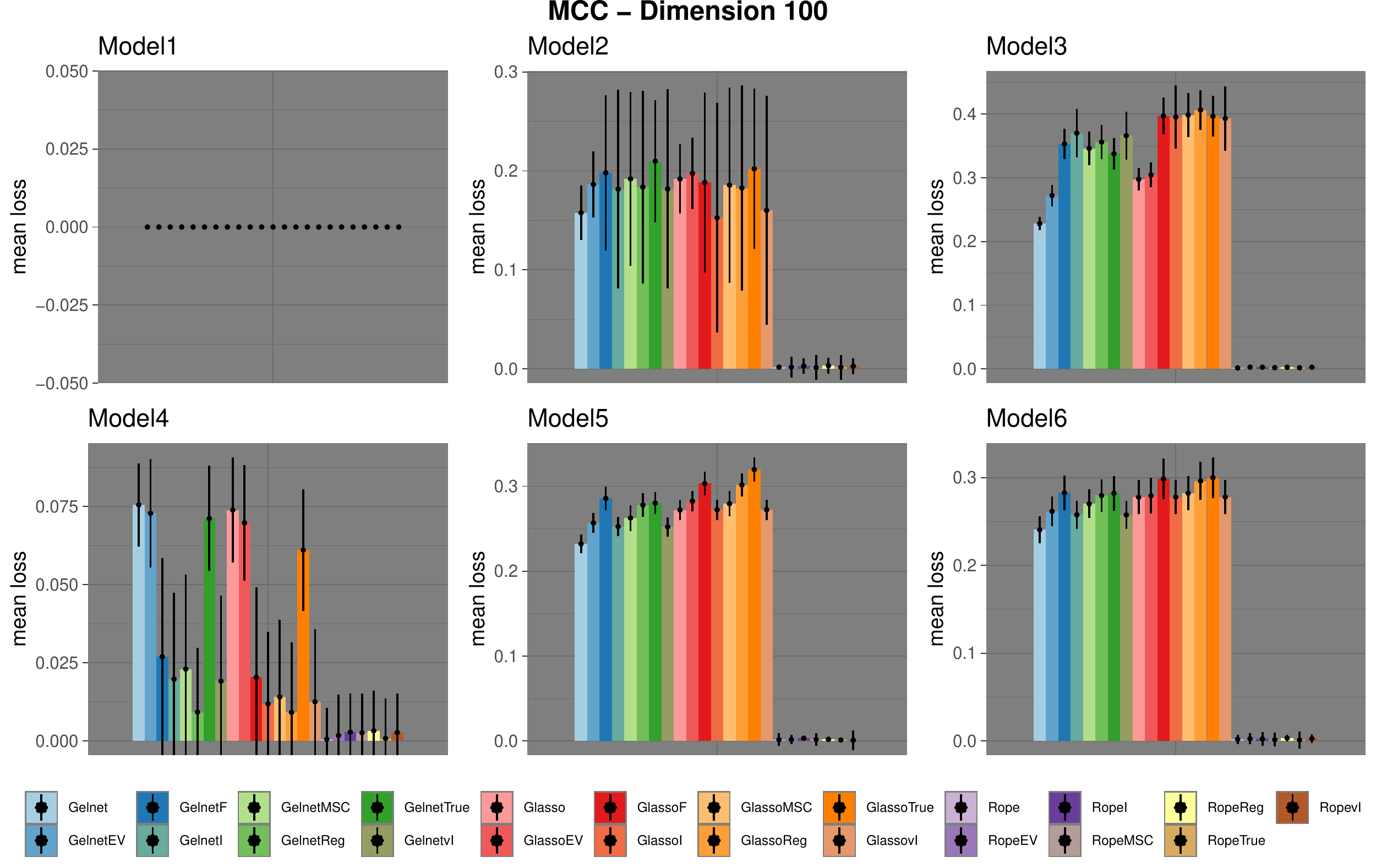}
  \caption{Summary of MCC score for the different models and different methods based on 100 replications. The columns along the small black dots indicate mean scores. The bars on the top of each column show the standard deviations (mean $\pm$ SD). Plot layout is taken from \cite{Kuismin_ROPE}.}
    \label{fig:MCC}
\end{figure}

\paragraph{Details on the methods used.}
Note that due to the $L_2$-penalty term, the scaling of the variables matters. Thus, as input $\textbf{S}$, we recommend to use the sample correlation rather than covariance in practice. For the Graphical Elastic Net we fixed $\alpha=\tfrac{1}{2}$ for the Elastic Net penalty in equation~\eqref{problem_setup} whenever using it. The simulations were done with and without targets as well as with and without penalizing the diagonal. Notice that in our case the targets are diagonal matrices and not penalizing the diagonal automatically results in no target. The following abbreviations are used in the figures, and the target matrices mentioned here are described in section \ref{sec:target_types}.
\begin{itemize}
    \item F: Setting the penalization parameter to FALSE, i.e. no penalization of the diagonal.%
    \item True: Using the \textit{True Diagonal} target matrix.
    \item I: Using the \textit{Identity} target matrix.
    \item vI: Using the $v$-\textit{Identity} target matrix.
    \item EV: Using the \textit{Eigenvalue} target matrix.
    \item MSC: Using the \textit{Maximal Single Correlation} target matrix.
    \item Reg: Using the \textit{Nodewise Regression} target matrix.
\end{itemize}

\paragraph{Simulation results.}
The results are displayed in Figures~\ref{fig:KL}, \ref{fig:L2}, \ref{fig:SP}, \ref{fig:F1} and \ref{fig:MCC}. As the \DPGN\ algorithm leads to indistinguishable results compared to the \GN\ algorithm, we left it out from the figures and also in the discussion below. Various performance measures might yield different ranking of the methods, such that it is hard to draw an overall conclusion on which method is superior as well as a general recommendation on which method or target matrices to use. Nonetheless, we highlight here a few observations based on the results shown in the figures:
\begin{itemize}
    \item Not penalizing the diagonal of $\Ma{\Theta}$ leads to smaller losses than penalizing with zero as a target.
    \item While the message from \cite{Kuismin_ROPE}, that \RP\ with target performs better than \GL\ with no target, in some of the models holds true, there are some models opposing this statement in general.
    \item In general, it is of advantage to include a target in the algorithms, as each method performs better if a suitable target is chosen.
    \item The question of how to determine the optimal target matrix is still open. Our simulations show that the \textit{True Diagonal} target matrix performs best in terms of L2-Loss and SP-Loss and thus approaching the diagonal of the underlying precision matrix is desired. However, the \textit{True Diagonal} target is not necessarily the winner in terms of \textit{Kullback-Leibler loss}.
    \item There is a systematic bias for the estimated diagonal entries of the covariance matrices if diagonal entries are penalized. For example, for the standard \GL\ algorithm with zero as target matrix, the diagonal of $\hat{\mathbf{W}}$ %
    is set as $\hat{w}_{i,i}=s_{i,i}+\lambda$ (if the diagonal is penalized). Similarly, bias of the diagonal entries occurs also for Elastic Net penalties and target matrices. In all these cases, one could try re-scaling the matrix, such that the $\hat{w}_{i,i}=s_{i,i}$. All statements above would still hold, but are often less clearly visible.
\end{itemize}

\subsection{Conclusions on empirical performances}
Including a reasonable target matrix seems to improve the estimation and is thus recommended. The intuition that the true diagonal (or something close to it) is the best diagonal target does not necessarily hold. Furthermore, there is no overall winner. Depending on the underlying precision matrix and the considered performance measure, different target types might be better suited. Hence, we recommend to explore different target types as a tool to gain better insight into the data. The benefits of Elastic Net penalties are not clearly visible in the considered simulations. Bigger differences are expected for highly correlated variables, analogous to the findings in regression by \cite{elastic_net}, where Elastic Net could potentially lead to more stable estimates.

\section{Computational times} \label{chp:comptime}
Our implementations in the \oursoftware\ \Rsoftware-package build upon the \textsf{glasso} \citep{glasso_package} and the \textsf{dpglasso} \cite{dpglasso_package} \textsf{R}-packages by suitably modifying them. Additionally, we also utilized a faster implementation from the
\textsf{glassoFast} \Rsoftware-package \citep{glassofast_package, glassofast_paper} that avoids unnecessary copying of subsets of the working covariance matrix when setting up row and column updates, which causes some inefficiency for the \textsf{glasso} package. Moreover, whenever possible, we even improve on the original versions. For example, to achieve further speed-ups, we incorporate block diagonal screening \citep{Mazumder2012_connected_components, Witten2011_connected_components} which was missing in the \textsf{dpglasso} and \textsf{glassoFast} packages. Compared to the \textsf{dpglasso} package, our \textsf{dpgelnet} implementation relies even more on \textsf{Fortran}, and it avoids unnecessary copying of working covariance matrices (in the style of the \textsf{glassoFast} package), which are beneficial for its speed.

\subsection{Comparing to existing Graphical Lasso implementations}
Figure~\ref{glasso_time_benchmark} compares the computational speed for Graphical Lasso estimation problems (i.e.,~$\alpha=1$) using the implementations from the \oursoftware, \textsf{glasso} and \textsf{glassoFast} \Rsoftware-packages. Our \textsf{gelnet} implementation in the \oursoftware\ package is implemented similar to the \textsf{glassoFast} package and hence, similarly fast. However, we additionally included the block-diagonal screening rule of \cite{Mazumder2012_connected_components, Witten2011_connected_components}. As shown on the right plot of Figure~\ref{glasso_time_benchmark}, in case the problem can be decomposed into connected components (see also Lemma~\ref{thm:connected_components}), our implementation can be considerably faster. There might be problems though where everything (or at least a predominant majority of the nodes) belong to the same connected components, see the left of Figure~\ref{glasso_time_benchmark}. In such cases our implementation tends to be slightly slower compared to the \textsf{glassoFast} package, because checking for the connected components has to be done and additionally slightly more computations are carried out because our implementation can handle general Elastic Net penalties rather than only fine tuned computations for the Graphical Lasso case of $\alpha=1$.

\begin{figure}[H]
	\centering
  \includegraphics[width=0.95\textwidth]{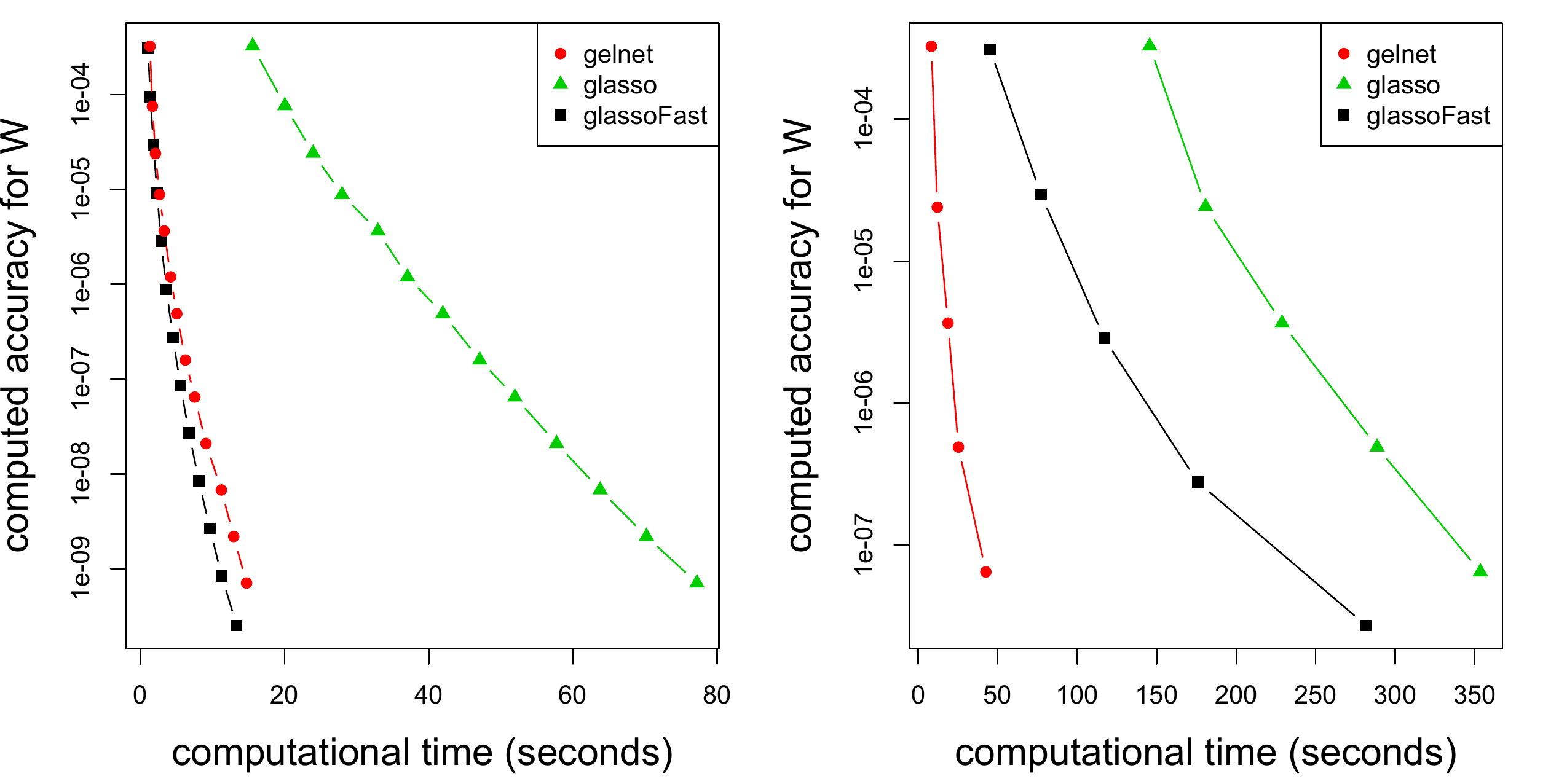}
  \caption{Computed accuracy (maximal absolute error of the off-diagonal entries) for the working covariance matrix W (on a logarithmic scale) vs.~average computational times of $10$ runs each. On the left the input correlation matrix is the one from the ER dataset ($p=692$, available in the \textsf{QUIC} \Rsoftware-package of  \citealp{QUIC_Hsieh}). On the right analogous results are shown for a $p=3460$-dimensional problem composed of five blocks of size $692$ each, i.e., the input here is a block diagonal matrix with five blocks, where the correlation matrix of the ER dataset was used for each of the five blocks. In all cases a standard Graphical Lasso estimator (i.e.,~$\alpha=1$) with tuning parameter $\lambda = 0.3$ was computed using the implementations from the \oursoftware, \textsf{glasso} and \textsf{glassoFast} \Rsoftware-packages.}
  \label{glasso_time_benchmark}
\end{figure}

\subsection{Computational cost for target matrices and Elastic Net penalties}

We introduced new methodology for sparse precision matrix estimation and in the following we would like to demonstrate efficient implementations also for these new proposals. 
To measure the computational time of the different methods the \textsf{FHT} data set available e.g.~in the \textsf{R}-package \textsf{gcdnet} \citep{gcdnet_package}, is used. It contains $n=50$ observations and $p=100$ variables, for which we calculate the empirical correlation matrix $\textbf{S} \in \mathbb{R}^{p \times p}$. %
For a fixed penalty parameter $\lambda$ each method computes the solution 20 times without warm start and 20 times with a warm start coming from the solution with a similar parameter $\lambda$. The average computational times are displayed in Figure~\ref{Bench}. The abbreviations for the algorithms are similar to those of section~\ref{chp:simulation} with the small additions that the number in the brackets denotes the value of $\alpha$ and IRfGlassoT is the Iterative Ridge for Glasso with a target matrix as proposed and implemented by \cite{Wieringen_iterative}. Table~\ref{tab:abbreviation} summarizes the abbreviations of the methods and the availability of the implementations. The target matrix $\textbf{T}$ used in these simulation is always the \textit{Identity} matrix.

The plots on the left of Figure~\ref{Bench} are based on using values $\lambda$, such that only one connected component is present. Note that the required $\lambda$ differs for various $\alpha$ values in the Elastic Net penalty (see Lemma~\ref{thm:connected_components}). \GN\ and \DPGN\ are competitive with \GL\ and \DPGL\ both in the cold (on top) and in the warm start case (on the bottom) for $\alpha=0.5$ and $\alpha=1$. For small $\alpha$ they slow down in speed and the closed form \RP\ for $\alpha=0$ is much faster. However, note here, that the closed form for \RP\ does not allow entry-wise differing $\lambda$ values which would be possible using our iterative \GN\ and \DPGN\ algorithms with $\alpha = 0$. To further illustrate the efficiency of the implementations, the gains of using connected components are shown on the right panel of Figure~\ref{Bench}. The parameter $\lambda$ is chosen such that the largest connected component is of size 50, whenever this is possible, i.e. for all algorithms except \RP, \GN ($\alpha=0$) and  \DPGN ($\alpha=0$). The timings show that instead of using \DPGL, where only the inner loop of coordinate descent is implemented in Fortran, it is recommended to use \DPGN\ with $\alpha=1$, where both the outer and inner loop are implemented in Fortran, the more efficient general implementation in the style of the \textsf{glassoFast} package is included, and the check for connected components is performed prior to computations in order to gain further computational efficiency. Besides the speedup resulting from solving the problem within several smaller blocks due to the results on connected components, this setup is also much more sparse (due to a larger value of the tuning parameter $\lambda$) for Elastic Net based estimators, which leads to considerable speedups compared to the left panel. In particular, Elastic Net based sparse estimators (with $\alpha = 0.5$ and $\alpha=1$) are even faster in this case than the closed form solution for \RP.

The only competitor that is able to incorporate a target matrix into the Graphical Lasso is the Iterative Ridge for Glasso of \cite{Wieringen_iterative}. The available implementation of the latter approach is limited to $\alpha=1$, and what is even more critical is that its speed is orders of magnitude behind all of our presented algorithms. In particular the Iterative Ridge for Glasso approach is much slower than our \GN\ algorithm with target matrices, with the latter also enabling targets in the full range of $\alpha\in [0,1]$. Lastly, note again that the \textsf{glassoFast} implementation of the \GL\ algorithm is considerably faster than the one from the \textsf{glasso} package. Our \textsf{gelnet} implementation (for $\alpha=1$) is approximately as fast as the \textsf{glassoFast}.

\begin{table}[H]
\caption{\label{tab:abbreviation} Abbreviations of the implementations used for Figure~\ref{Bench}}
\centering
\fbox{\begin{tabular}{cccc} 
\em abbreviation    & \em problem, algorithm  & \em implementation (function, package) \\  
\hline
Rope            & $\alpha = 0$, closed form         & rope, \oursoftware     \\
Glasso          & $\alpha = 1$, \GL\                & glasso, \textsf{glasso}\\
GlassoFast      & $\alpha = 1$, \GL\                & glassoFast, \textsf{glassoFast}\\
Gelnet(1)       & $\alpha = 1$, \GN\                & gelnet, \oursoftware  \\
Gelnet(0.5)     & $\alpha = 0.5$, \GN\              & gelnet, \oursoftware     \\
Gelnet(0)       & $\alpha = 0$, \GN\                & gelnet, \oursoftware     \\
DPGlasso        & $\alpha = 1$, \DPGL\              & dpglasso, \textsf{dpglasso}    \\
DPGelnet(1)     & $\alpha = 1$, \DPGN\              & dpgelnet, \oursoftware   \\
DPGelnet(0.5)   & $\alpha = 0.5$, \DPGN\            & dpgelnet, \oursoftware     \\
DPGelnet(0)     & $\alpha = 0$, \DPGN\              & dpgelnet, \oursoftware      \\
GelnetT(1)      & $\alpha = 1$ with target, \GN\    & gelnet, \oursoftware     \\
GelnetT(0.5)    & $\alpha = 0.5$ with target, \GN\  & gelnet, \oursoftware     \\
IRfGlassoT      & $\alpha = 0.1$ with target,       & supplement of \cite{Wieringen_iterative}\\
                & Iterative Ridge for Glasso with target&                        \\
\end{tabular}}
\end{table}

\begin{figure}[H]
	\centering
  \includegraphics[width=1.0\textwidth]{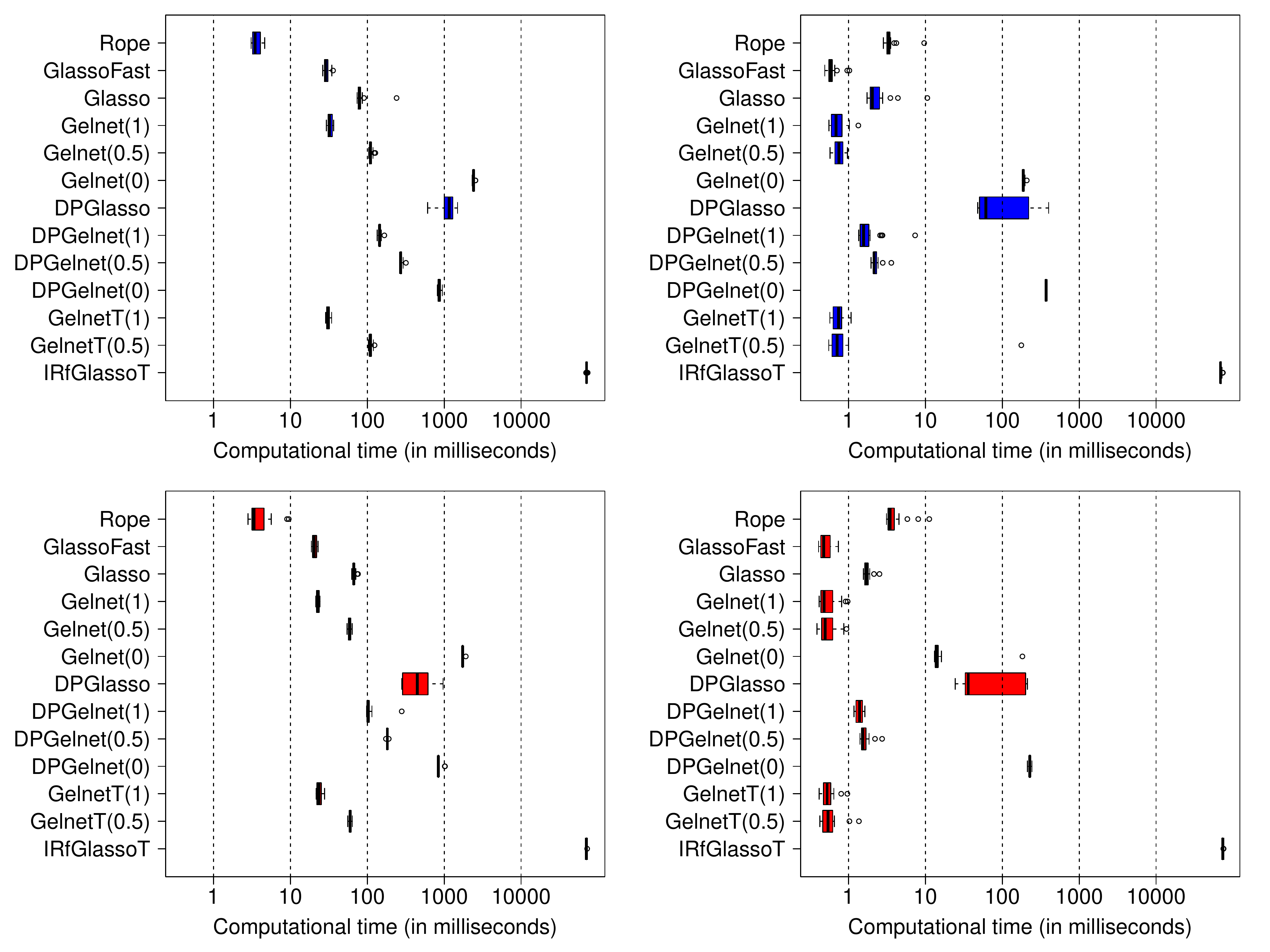}
  \caption{Average computation times for the $100$-dimensional correlation matrix from the \textsf{FHT} dataset with a fixed penalty parameter with a cold starts (on the top in blue) or a warm start, coming from a similar penalty parameter (on the bottom in red). The plots on the right used $\lambda$ such that $\mathbf{\Theta}$ disconnected into connected components of maximum size 50, whereas $\lambda$ of the plots on the left is chosen such that $\mathbf{\Theta}$ stayed fully connected. For the explanations of the names, see Table~\ref{tab:abbreviation}.}
  \label{Bench}
\end{figure}

\subsection{Conclusions on computational times}
By combining the best parts of various available implementations, our software often improves previous DP-Graphical Lasso implementations considerably, and it is also as fast, or in some scenarios even faster than existing efficient Graphical Lasso implementations. Moreover, our software generalizes beyond the case $\alpha=1$ to incorporate Elastic Net penalties and the Graphical Elastic Net variant also allows to include a diagonal target matrix. Computation with Elastic Net penalties takes typically slightly longer than with the $L_1$-penalty only. Our implementation for target matrices is also efficient, and it is orders of magnitude faster than the competing but ad-hoc Iterative Ridge algorithm of \cite{Wieringen_iterative}.

\section{Conclusions} \label{chp:conclusions}

We consider Elastic Net type penalization for precision matrix estimation based on the Gaussian log-likelihood. To resolve this task, we propose two novel algorithms \GN\ and \DPGN. They substantially build on earlier algorithmic work for the \GL\ for precision matrix estimation (\citealp{glasso, Mazumder2012_dpglasso}), the Elastic Net for regression \citep{elastic_net} and the inclusion of target matrices to encode some prior information towards which the estimator is shrunken (\citealp{Wieringen_Peeters, Kuismin_ROPE}). Advantages of our new work include methodological extensions enabling higher flexibility for data analysis, the important option for including diagonal target matrices into the estimation methods, user-friendly software, efficient implementation and the possibility to choose different methodological versions within the same software package.

\DPGN\ optimizes over the desired precision matrix in each block update, whereas \GN\ works with the covariance matrix, the inverse of the precision matrix. There is no overall advantage of using one over the other algorithm and computational performance gains depend on the true underlying signal. All our algorithms are implemented efficiently in the \Rsoftware-package \oursoftware\ using a combination of \Rsoftware\ and Fortran code. They are competitive in terms of computational times with competing methods.

There seems to be no overall winner in terms of statistical performance. The inclusion of an $L_2$-norm into the Elastic Net penalty encourages additional stability, especially in presence of highly correlated variables. Moreover, our simulations show that target matrices and not penalizing the diagonal can be powerful tools to improve the estimation of precision matrices. Our new software in the \Rsoftware-package \oursoftware\ and its algorithmic methodology provide a unifying tool to use various modifications of the popular \GL\ algorithm.

\section*{Acknowledgement}
Solt Kov\'acs and Peter B\"uhlmann have received funding from the European Research Council (ERC) under the European Union's Horizon 2020 research and innovation programme (Grant agreement No. 786461 CausalStats - ERC-2017-ADG). 

\bibliographystyle{apalike}
\bibliography{ref}

\appendix
\section{Appendix}
\label{Appendix}
\subsection{Details on target matrices}
\label{app:target}

We discuss the inclusion of a positive semi-definite (diagonal) target matrix $\textbf{T}$ into the Graphical Elastic Net problem. We aim to solve problem \eqref{problem_setup} which we recall here:
\begin{equation*}
       \hat{\mathbf{\Theta}}(\lambda, \alpha, \textbf{T}) = 
    \underset{\mathbf{\Theta} \succ 0 }{\text{argmin}}   \{ -\text{log det} \Ma{\Theta} + \text{tr} (\textbf{S} \Ma{\Theta}) + \lambda (\alpha \norm{\Ma{\Theta}- \textbf{T}}_1+\tfrac{1-\alpha}{2} \norm{\Ma{\Theta}-\textbf{T}}_2^2) \} .
\end{equation*}
This optimization problem is difficult to solve efficiently in general. We propose to simplify the problem by considering only positive semi-definite \textbf{diagonal} target matrices. When considering diagonal target matrices with non-negative entries, the normal equations for the diagonal entries are different, while for the other entries the normal equations remain the same. For each diagonal entry three cases can occur, leading to the following normal equations:\\ \\
Case 1: $\theta_{22} > t_{22}$,  then
\begin{equation*}
 w_{22} = s_{22} + \lambda \alpha + \lambda (1-\alpha) (\theta_{22}-t_{22})  .
\end{equation*}
Case 2: $\theta_{22} = t_{22}$,  then
\begin{equation*}
 w_{22} = s_{22} + \lambda \alpha u, \text{ where } u \in [-1,1]. 
\end{equation*}
Case 3: $\theta_{22} < t_{22}$,  then
\begin{equation*}
 w_{22} = s_{22} - \lambda \alpha + \lambda (1-\alpha) (\theta_{22}-t_{22}) .
\end{equation*} 
Note that in the traditional case when zero is the target case~1 occurs always. In the general case, however, one cannot automatically update diagonal elements based on case~1. The difficulty is that one does not know in advance the precision matrix $\mathbf{\Theta}$ and hence, whether updates corresponding to case 1, 2 or 3 would be suitable for the diagonal entries. This decision has to be made ``on the fly''. In the following we derive the technical modifications on how to modify the \GN\ algorithm. Overall, in contrast to the algorithm without target, the initialization and updating for the diagonal entries are different.

\paragraph{Special Case.} 
In the special case, where all off-diagonal elements of $\Ma{S}$ are less or equal to $\lambda\alpha$, the solutions $\Ma{\Theta}$ and $\Ma{W}$ are by the results on connected components (see Theorem \ref{thm:connected_components}) diagonal matrices and thus $\theta_{22} = \tfrac{1}{w_{22}}$ for each diagonal entry. If $\alpha=1$ then $w_{22}=s_{22}+\lambda$ and thus $\theta_{22}=\frac{1}{s_{22}+\lambda}$. We consider now $\alpha \neq 1$. By the condition that $\theta_{22}>0$ and $w_{22}>0$ only one case per entry can arise.\\ \\
Case 1: $\theta_{22} > t_{22} \geq 0$ and $\frac{1}{\theta_{22}}=w_{22} = s_{22} + \lambda \alpha + \lambda (1-\alpha) (\theta_{22}-t_{22}) > s_{22} + \lambda \alpha$. 
Therefore $ t_{22} < \theta_{22} < \frac{1}{s_{22} + \lambda \alpha}$. 
In other words case 1 is fulfilled if $t_{22} \in [0, \frac{1}{s_{22} + \lambda \alpha})$.
\\ \\
Case 2: $\theta_{22} = t_{22} \geq 0$ and $\frac{1}{t_{22}}=\frac{1}{\theta_{22}}=w_{22} = s_{22} 
 + \lambda \alpha u, \text{ where } u \in [-1,1]$.
Assume $\lambda \alpha < s_{22}$. Then $t_{22}= \frac{1}{s_{22} 
 + \lambda \alpha u}$
, which is equivalent to $t_{22} \in  [\frac{1}{s_{22} + \lambda \alpha}, \frac{1}{s_{22} - \lambda \alpha}]$. 
On the other hand if $\lambda \alpha \geq s_{22}$, then $t_{22} \in  [\frac{1}{s_{22} + \lambda \alpha}, \infty)$ follows, since the target cannot be negative.
\\ \\
Case 3: $\theta_{22} < t_{22} \geq 0$ and $\frac{1}{\theta_{22}}=w_{22} = s_{22} - \lambda \alpha + \lambda (1-\alpha) (\theta_{22}-t_{22}) < s_{22} - \lambda \alpha$. 
Therefore, if $\lambda \alpha \geq s_{22}$ then case 3 will never be fulfilled. Else $ t_{22} > \theta_{22} > \frac{1}{s_{22} - \lambda \alpha}$. 
In other words case 3 is fulfilled if $\lambda \alpha < s_{22}$ and $t_{22} \in (\frac{1}{s_{22} - \lambda \alpha}, \infty)$.
\\ \\
Hence, for this special case, one can get the exact values of $\theta_{22}$ by first determining with $t_{22}$ which case occurs. If $\theta_{22} \neq t_{22}$ a quadratic equation has to be solved. The values for $w_{22}$ follow by $w_{22}=\frac{1}{\theta_{22}}$. \\ \\ 
Case 1: 
\begin{equation}\label{case1}
    \theta_{22} = \frac{-(s_{22}+\lambda \alpha - \lambda (1-\alpha) t_{22}) + \sqrt{(s_{22}+\lambda \alpha - \lambda (1-\alpha) t_{22})^2 + 4\lambda(1-\alpha) } }{2 \lambda (1-\alpha)}
\end{equation}
Case 2: \hspace{0.74cm} $\theta_{22}=t_{22}$
\\ \\
Case 3: 
\begin{equation*}
    \theta_{22} = \frac{-(s_{22} - \lambda \alpha - \lambda (1-\alpha) t_{22}) + \sqrt{(s_{22} - \lambda \alpha - \lambda (1-\alpha) t_{22})^2 + 4\lambda(1-\alpha) } }{2 \lambda (1-\alpha)}
\end{equation*}

\paragraph{General Case.}
Note that in the general case, where the off-diagonals of $\mathbf{S}$ are no longer smaller or equal to $\lambda \alpha$ the special case still gives some intuition. Namely, for small values in the target case 1 is present and the larger the values of the target get, the more likely case 3 is the suitable one. As we assume that the target matrix is diagonal, and hence has non-zero elements only at the diagonal entries, the core of the \GN\ algorithm stays the same. The changes to be made are the following:
\begin{itemize}
    \item Initialize such that both $\mathbf{W}$ and $\mathbf{\Theta}$ are positive semi-definite and preferably the suitable diagonal case is chosen.
    \item After solving the quadratic programs, update such that the diagonal entries can switch cases if needed.
\end{itemize}

\paragraph{Updates.} The updates of  $\hat{w}_{22}$ and $\hat{\theta}_{22}$ (Steps 2 d and f in Algorithm \ref{alg:algorithm3}) differ from the \GN\ algorithm without target. To ensure that $\hat{w}_{22}$  and  $\hat{\theta}_{22}$ fall into the same case a simultaneous update of both is needed. In practice, the following updates work for realistic targets, but may not converge for targets with very large diagonal entries. However, such targets with overly large diagonal entries are typically anyway not desirable from a statistical perspective.\\ \\
Using the relation from Step 2d in Algorithm \ref{alg:algorithm3} $w_{22}$ can be expressed as $w_{22}=\frac{1}{\theta_{22} + w_{12} \beta}$. Substituting this expression into the diagonal normal equation leads to    \begin{equation} \label{target:equ}
    \lambda \alpha \text{ sgn}(\theta_{22}-t_{22})=\tfrac{1}{\theta_{22}} + w_{12} \beta - s_{22} - \lambda (1-\alpha) (\theta_{22}-t_{22}) .
\end{equation}
Define the function $F$, which represents $\lambda \alpha \text{ sgn}(\theta_{22}-t_{22})$ as:
\begin{equation*}
    F(\theta_{22}) \coloneqq \tfrac{1}{\theta_{22}} + w_{12} \beta - s_{22} - \lambda (1-\alpha) (\theta_{22}-t_{22}) .
\end{equation*}
Note the range of values for $F$ is $[-\lambda \alpha, \lambda \alpha]$. Define for $\theta_{22}=t_{22}$ the function value $F_t$ as $F_t \coloneqq F(t_{22})=\tfrac{1}{t_{22}} + w_{12} \beta - s_{22}.$ 
Since $F$ is strictly decreasing in $\theta_{22} \geq 0 $, we have that:
 $F(\theta_{22}) < F_t$ for $\theta_{22} > t_{22}$  and $F(\theta_{22}) > F_t$ for $\theta_{22} < t_{22}$. \\
By the representation above we need for $\theta_{22} > t_{22}$ that $F(\theta_{22}) = \lambda \alpha$ and for $\theta_{22} < t_{22}$ that $F(\theta_{22}) = -\lambda \alpha$.\\
Considering $F_t \in [-\lambda \alpha, \lambda \alpha]$, we show by contradiction that $\theta_{22}=t_{22}$. \\
Assume $\theta_{22} > t_{22}$ then $\lambda \alpha = F(\theta_{22}) < F_t \leq \lambda \alpha$. \\
On the other hand if $\theta_{22} < t_{22}$ then $-\lambda \alpha = F(\theta_{22}) > F_t \geq -\lambda \alpha$.
\\
By similar arguments one can show that if
$F_t > \lambda \alpha$, we have that $\theta_{22} > t_{22}$ and lastly if $F_t  < -\lambda \alpha$, we have that $\theta_{22} < t_{22}$.
\\
Therefore, the value of $F_t$ determines the case to be considered and can be seen as a ``test''. We need to solve equation \eqref{target:equ} in $\theta_{22}$ dependent on the case chosen.\\ \\
Case 1:  Solve the following for $\theta_{22}$:\begin{equation*}
    0=\lambda(1-\alpha) \theta_{22}^2 + (s_{22}+\lambda \alpha - \lambda (1-\alpha) t_{22} - w_{12} \beta ) \theta_{22} -1.
\end{equation*}
Then update $w_{22} = s_{22}+ \lambda \alpha + \lambda (1-\alpha) (\theta_{22}-t_{22})$. \\ \\
Case 2: $\theta_{22} = t_{22}$, $w_{22}=s_{22}+F_t$. \\ \\
Case 3:  Solve the following for $\theta_{22}$:\begin{equation*}
    0=\lambda(1-\alpha) \theta_{22}^2 + (s_{22}-\lambda \alpha - \lambda (1-\alpha) t_{22} - w_{12} \beta ) \theta_{22} -1.
\end{equation*}
Then update $w_{22} = s_{22} - \lambda \alpha + \lambda (1-\alpha) (\theta_{22}-t_{22})$.

\paragraph{Initialization.} 
Assume $\alpha>0$. As in \GN, the diagonal entries of $\mathbf{\Theta}$ are produced first and then $\mathbf{W}$ is taken as $\mathbf{W}= \mathbf{S} + \lambda \alpha \mathbf{\Gamma} + \lambda (1-\alpha) (\mathbf{\Theta}-\mathbf{T})$. This time $\mathbf{\Gamma}$ is the diagonal matrix with $\gamma_{ii}=1$ if $t_{ii} \in [0, \frac{1}{s_{ii} + \lambda \alpha})$ and $0$ else. If $t_{ii} \in [0, \frac{1}{s_{ii} + \lambda \alpha})$ then define $\theta_{ii}$ as in \eqref{case1} else set $\theta_{ii}=t_{ii}$.

\paragraph{The Final Algorithm.} In the previous paragraphs we discussed the necessary technical modifications regarding updates and initialization for the \GN\ algorithm to incorporate diagonal target matrices. These are again summarized in Algorithm \ref{alg:gelnet_target}.
Note again that for very large entries in the target matrix (which are usually beyond what is reasonable from a statistical perspective) our chosen updates might face convergence issues when used with the \GN\ algorithm. For none of the targets we used in simulations (see section~\ref{sec:target_types}), we experienced convergence issues. These reasonable targets are ``conservative'',~i.e. typically not having overly large entries.

\begin{algorithm}[H]
    \caption{\GN\ algorithm with diagonal target matrix \bf{T}}
  \begin{algorithmic}[1]
    \State %
    Initialize $\Ma{\Theta}_{\text{init}}$ and $\Ma{W}_{\text{init}}$ as described in the \textbf{Initialization} paragraph above.
    \State Cycle around the columns repeatedly, performing the following steps till convergence:
\begin{algsubstates}
        \State Rearrange the rows/columns so that the target column is last (implicitly).
        \State Solve the Elastic Net regression problem \eqref{QP2} with coordinate descent to get $\hat{\Bo{\beta}}$. As warm start for $\Bo{\beta}$ use the solution from the previous round for this row/column.
        \State Update $\hat{\Ma{w}}_{12} = \textbf{W}_{11} \hat{\Bo{\beta}}$, $\hat{\Ma{w}}_{21}=\hat{\Ma{w}}_{12}^T$.
        \State Calculate the test statistic $F_t \coloneqq F(\theta_{22})$ as in the \textbf{Updates} paragraph to determine the case.
        \State Update  $\hat{\theta}_{22}$ according to the case.
        \State Update $\hat{\Bo{\theta}}_{12}=- \hat{\theta}_{22} \hat{\Bo{\beta}}$, $\hat{\Bo{\theta}}_{21}=\hat{\Bo{\theta}}_{12}^T$.
        \State Update $\hat{w}_{22}= \text{s}_{22}+\lambda \alpha+ \lambda (1-\alpha) \hat{\theta}_{22}$.
    \end{algsubstates}
  \end{algorithmic} 
  \label{alg:gelnet_target}
\end{algorithm}

\end{document}